\documentclass[11pt, a4paper, reqno]{amsart}

\usepackage[foot]{amsaddr} 

\usepackage{amsfonts}
\usepackage{amssymb}
\usepackage[dvips]{graphics}
\usepackage{epsfig}
\usepackage{color}
\usepackage[pdftex]{hyperref}
\usepackage{fullpage}
\usepackage{Preamble}

\usepackage{xcolor}
\hypersetup{
	colorlinks,
	linkcolor={black!30!blue},
	citecolor={black!30!blue},
	urlcolor={black!30!blue}
}

\begin{document}

\title{The double semion state in infinite volume}
\author{Alex Bols${}^1$}
\address{${}^1$ California Institute of Technology, Pasadena, CA 91125, USA}
\email{\href{mailto:abols@caltech.edu}{abols@caltech.edu}}

\author{Boris Kj\ae r${}^2$}
\address{${}^2$ QMATH, Department of Mathematical Sciences, University of Copenhagen, Universitetsparken 5, 2100 Copenhagen, Denmark}
\email{\href{mailto:bbk@math.ku.dk}{bbk@math.ku.dk}}

\author{Alvin Moon${}^3$}
\address{${}^3$ RAND Corporation. Santa Monica, CA. USA}
\email{\href{mailto:alvinm@rand.org}{alvinm@rand.org}}

%\date{\today}

\begin{abstract}
	We describe in a simple setting how to extract a unitary braided fusion category from a collection of superselection sectors of a two-dimensional quantum spin system, corresponding to abelian anyons. The structure of the unitary braided fusion category is given by F and R-symbols, which describe fusion and braiding of the anyons. We then construct the double semion state in infinite volume and extract the unitary braided fusion category describing its semion, anti-semion, and bound state excitations. We verify that this category corresponds to the representation category of the twisted quantum double $\caD^{\phi}(\Z_2)$.
\end{abstract}

\maketitle

%%%%%%%%%%%%%%%%%%%%%%%%%%%%%%%%%%%%%%%%%%%%%%%%%%%%%%%%%%%%%%%%%%%%%%%%%%%%%%%%%%%
%%%%%%%%%%%%%%%%%%%%%%%%%%%%%%%%%%%%%%%%%%%%%%%%%%%%%%%%%%%%%%%%%%%%%%%%%%%%%%%%%%%
%%%%%%%%%%%%%%%%%%%%%%%%%%%%%%%%%%%%%%%%%%%%%%%%%%%%%%%%%%%%%%%%%%%%%%%%%%%%%%%%%%%
\section{Introduction}

Gapped ground states of two dimensional quantum lattice systems may support anyonic exitations. Paradigmatic examples include Kitaev's quantum double models \cite{kitaev2003fault} and, more generally, the string net models of Levin and Wen \cite{levin2005string}. It is widely believed that the types of anyons supported by a given ground state are organized in a unitary braided fusion category, whose objects are the anyon types, and whose structure is described by F and R-matrices. 

Recent years have seen a lot of progress towards justifying this belief, by adapting the DHR analysis of superselection sectors in algebraic quantum field theory \cite{doplicher1971local, doplicher1974local, fredenhagen1989superselection, fredenhagen1992superselection, frohlich1990braid} to the setting of microscopic lattice systems \cite{naaijkens2011localized, cha2020stability, ogata2022derivation}. This body of work manages to associate to a large class of gapped ground states a strict braided $C^*$-tensor category whose objects are localized and transportable endomorphisms of the observable algebra, and shows that this category is a robust label for the gapped phase \cite{bachmann2012automorphic} to which the ground state belongs.

To obtain a unitary braided fusion category described in terms of F and R-symbols, one may skeletonise the braided $C^*$-tensor category refered to above, yielding in particular a microscopic definition of the F and R-symbols. We note that a microscopic definition of F and R-symbols has been given before in \cite{kawagoe2020microscopic}. Extracting the F and R-symbols from the DHR-type analysis has the advantage that it is clear that this data yields a robust label of gapped phases \cite{cha2020stability, ogata2022derivation}.

In section \ref{sec:braided fusion category} of the paper we review the construction of a braided $C^*$-tensor category from a given ground state under some simplifying assumptions. We assume in particular that the anyons we consider are abelian. By moreover assuming that our set of anyon types is finite and that each anyon has an anti-particle, we obtain by skeletonization a unitary braided fusion category with the anyons types as objects, and structure given by F and R-symbols that satisfy the pentagon and hexagon equations.

Section \ref{sec:double semion} is devoted to the construction and analysis of the double semion state \cite{levin2005string} in infinite volume. This is the simplest state that supports abelian anyons whose unitary braided fusion category has non-trivial F-symbols.

In the appendix we show various properties of the double semion state that are used in Section \ref{sec:double semion}. In particular, the appendix contains a proof that the double semion state is pure.

%%%%%%%%%%%%%%%%%%%%%%%%%%%%%%%%%%%%%%%%%%%%%%%%%%%%%%%%%%%%%%%%%%%%%%%%%%%%%%%%%%%
%%%%%%%%%%%%%%%%%%%%%%%%%%%%%%%%%%%%%%%%%%%%%%%%%%%%%%%%%%%%%%%%%%%%%%%%%%%%%%%%%%%
%%%%%%%%%%%%%%%%%%%%%%%%%%%%%%%%%%%%%%%%%%%%%%%%%%%%%%%%%%%%%%%%%%%%%%%%%%%%%%%%%%%
\section{Braided fusion category for abelian anyons} \label{sec:braided fusion category}

%%%%%%%%%%%%%%%%%%%%%%%%%%%%%%%%%%%%%%%%%%%%%%%%%%%%%%%%%%%%%%%%%%%%%%%%%%%%%%%%%%%
%%%%%%%%%%%%%%%%%%%%%%%%%%%%%%%%%%%%%%%%%%%%%%%%%%%%%%%%%%%%%%%%%%%%%%%%%%%%%%%%%%%
\subsection{Setup}

%%%%%%%%%%%%%%%%%%%%%%%%%%%%%%%%%%%%%%%%%%%%%%%%%%%%%%%%%%%%%%%%%%%%%%%%%%%%%%%%%%%
\subsubsection{Algebra of observables and state}
Consider a countable set $\Gamma \subset \R^2$ of sites in the plane. To each site $x \in \Gamma$ we associate an algebra $\caA_x \simeq \End(\C^d)$ for some fixed $d$. For any finite $X \subset \Gamma$ we set $\caA_X = \bigotimes_{x \in X} \caA_x$. If $X \subset Y$ are finite subset of $\Gamma$ then there is a natural norm-preserving inclusion $\caA_X \xhookrightarrow{} \caA_Y$ by tensoring with the identity. 

For any infinite subset $Y \subset \Gamma$ we have a local net of algebras $\caA_X$ with $X \subset Y$, whose union is $\caA_{Y, \loc}$, the algebra of local observables supported in $Y$. Its completion is $\caA_Y := \overline{\caA_{Y, \loc}}^{\norm{\cdot}}$, the algebra of quasi-local observables supported in $Y$, and we get inclusions $\caA_X \xhookrightarrow{} \caA_Y$ also for infinite $X\subset Y$. We write $\caA_{\loc} = \caA_{\Gamma, \loc}$ and $\caA = \caA_{\Gamma}$ for the algebra of all local and all quasi-local observables respectively. The \emph{support} of an observable $O \in \caA$ is the smallest set $X \subset \Gamma$ such that $O \in \caA_X$. 

Similarly, the support of an automorphism $w$ of $\caA$ is the smallest set $Y$ such that $w|_{\caA_{Y^c}}=\text{id}_{\caA_{Y^c}}$.  % By the commutant lemma $w(\caA_Y)\subset \caA_Y$.
Any subset $S \subset \R^2$ of the plane determines a subset $\overline S = S \cap \Gamma$ of $\Gamma$, and we denote $\caA_{S} := \caA_{\overline S}$.

A major role is played by \emph{cones}. The cone with apex at $a \in \R^2$, axis $\hat v \in \R^2$ of unit length, and opening angle $\theta \in (0, 2\pi)$ is
\begin{equation}
	\Lambda_{a, \hat v, \theta} := \{ x \in \R^2 \, | \, (x - a) \cdot \hat v > \norm{x-a} \cos (\theta/2)   \}.
\end{equation}

We will consider a pure state $\omega$ on $\caA$ with GNS representation $(\pi_1, \caH, \Omega)$. For any $X \subset \Gamma$ we put $\caR(X) := (\pi_1(\caA_X))''$, the von Neumann algebra associated to $X$. For $S \subset \R^2$ we also write $\caR(S) = \caR(\overline S)$. 
We note that if $\Lambda$ is a cone, then $\caR(\Lambda)$ is a factor (Theorem 5.2 of \cite{naaijkens2011localized}).

%%%%%%%%%%%%%%%%%%%%%%%%%%%%%%%%%%%%%%%%%%%%%%%%%%%%%%%%%%%%%%%%%%%%%%%%%%%%%%%%%%%
\subsubsection{Superselection sectors}

\begin{definition}
	An irreducible representation $(\pi, \caH)$ of $\caA$ is said to satisfy the superselection criterion w.r.t. $\pi_1$ if for any cone $\Lambda$ there is a unitary $U \in \caB(\caH)$ such that
	\begin{equation}
		U \pi(O) U^* = \pi_1(O) \quad \text{for all} \,\, O \in \caA_{\Lambda^c}
	\end{equation}
\end{definition}

If two representations $\pi, \pi'$ are unitarily equivalent then we write $\pi \simeq \pi'$.

We assume that we have a finite set of irreducible representations $\caO = \{ \pi_a \, | \, a \in I \}$ of $\caA$ indexed by a labeling. We assume moreover that $\pi_a \simeq \pi_b$ if and only if $a = b$, so all sectors in $\caO$ are truly distinct. Moreover, $1 \in I$ so that $\pi_1 \in \caO$. We call $\pi_1$ the \emph{vacuum sector}.

We will now make some additional assumptions on these sectors which will in particular imply that they satify the superselection criterion w.r.t. $\pi_1$. These assumptions are not generically satisfied by gapped ground states, but they are satisfied for a large class of models, for example all abelian string net models \cite{levin2005string} including the double semion model studied below.

\begin{assumption} \label{ass:local string operators}
	For any cone $\Lambda$  and any $a \in I$ there is an automorphism $w_{a, \Lambda}$ supported on $\Lambda$ such that
	\begin{equation}
		\pi_a \simeq \pi_1 \circ w_{a, \Lambda}.
	\end{equation}
	In particular, we take $w_{1, \Lambda} = \id$ for all cones $\Lambda$.
\end{assumption}

The following assumption says that the anyons we study are abelian.
\begin{assumption} \label{ass:abelian fusion}
	For any $a, b \in I$ there is a unique $c \in I$ such that for any two cones $\Lambda_1, \Lambda_2$ we have $\pi_c \simeq \pi_1 \circ w_{a, \Lambda_1} \circ w_{b, \Lambda_2}$. We write $c = a \times b$.
\end{assumption}

The following assumption says that each anyon has an antiparticle.
\begin{assumption} \label{ass:existence of antiparticles}
	For each $a \in I$ there is an $a^* \in I$ such that $a \times a^* = 1$.
\end{assumption}

The final assumption is of a technical nature, it plays an important role in constructing a tensor category in Section \ref{subsec:braided tensor category}.
\begin{assumption} \label{ass:locality of intertwiners}
	Assumption \ref{ass:local string operators} implies that if $\Lambda_1, \Lambda_2 \subset \Lambda$ are cones, then $\pi_1 \circ w_{a, \Lambda_1} \simeq \pi_1 \circ w_{a, \Lambda_2}$. We assume that any unitary $V \in \caB(\caH)$ implementing this equivalence belongs to the von Neumann algebra $\caR(\Lambda)$.

	Similarly, it follows from assumption \ref{ass:abelian fusion} that $\pi_1 \circ w_{a, \Lambda} \circ w_{b, \Lambda} \simeq \pi_1 \circ w_{a \times b, \Lambda}$. We assume that any unitary $V \in \caB(\caH)$ implementing this equivalence belongs to the von Neumann algebra $\caR(\Lambda)$.
\end{assumption}
This assumption is implied by Haag duality, and can also be proven directly if the automorphisms $w_{a, \Lambda}$ have a lot of structure, for example for the Toric code \cite{naaijkens2011localized} and for the double semion model treated below.

Assumptions \ref{ass:local string operators}-\ref{ass:existence of antiparticles} have a few elementary but important consequences.
\begin{lemma} \label{lem:sectors satisfy superselection criterion}
	The representations $\pi_a$, $a \in I$ satisfy the superselection criterion w.r.t. $\pi_1$.
\end{lemma}

\begin{proof}
	Fix a cone $\Lambda$. By assumption \ref{ass:local string operators} there is an automorphism $w_{a, \Lambda}$ supported in $\Lambda$ such that $\pi_1 \circ w_{a, \Lambda} \simeq \pi_a$. \ie there is a unitary $U \in \caB(\caH)$ such that
	\begin{equation}
		U \pi_a(O) U^* = \pi_1( w_{a, \Lambda}(O) )
	\end{equation}
	for all $O \in \caA$. Since $w_{a, \Lambda}$ is supported in $\Lambda$, we have $w(O) = O$ for $O \in \caA_{\Lambda^c}$ and therefore
	\begin{equation}
		U \pi_a(O) U^* = \pi_1(O) \quad \text{for all} \,\, O \in \caA_{\Lambda^c}.
	\end{equation}
\end{proof}

\begin{lemma} \label{lem:I is abelian group}
	The binary operation $\times : I \times I \rightarrow I$ makes $I$ into an abelian group with unit $1$ and inverse $a^{-1} = a^*$.
\end{lemma}

\begin{proof}
	We first show that $\times$ is abelian. Take $a, b \in I$. Assumption \ref{ass:abelian fusion} says that for any two cones $\Lambda_1, \Lambda_2$ there are automorphisms $w_{a, \Lambda_1}$, $w_{b, \Lambda_2}$ such that $\pi_{a \times b} \simeq \pi_1 \circ w_{a, \Lambda_1} \circ w_{b, \Lambda_2}$. Exchanging the roles of $a$ and $b$ and of $\Lambda_1$ and $\Lambda_2$ we have $\pi_{b \times a} \simeq \pi_1 \circ w_{b, \Lambda_2} \circ w_{a, \Lambda_1}$. If we now take $\Lambda_1$ and $\Lambda_2$ to be disjoint then certainly $w_{a, \Lambda_1} \circ w_{b, \Lambda_2} = w_{b, \Lambda_2} \circ w_{a, \Lambda_1}$ and therefore $\pi_{a \times b} \simeq \pi_{b \times a}$. But we assumed that two representations in $\caO$ are unitarily equivalent only if they are the same, so $a \times b = b \times a$.

	We now show that $1$ is the identity for the product $\times$. Fix cones $\Lambda_1$ and $\Lambda_2$. By assumptions \ref{ass:local string operators} and \ref{ass:abelian fusion} there are automorphisms $w_{1, \Lambda_1} = \id$ and $w_{a, \Lambda_2}$ such that $\pi_{1 \times a} \simeq \pi_1 \circ \id \circ w_{a, \Lambda_2} = \pi_1 \circ w_{a, \Lambda_2} \simeq \pi_{a}$, hence $1 \times a = a$. We already know that $\times$ is abelian, so also $a \times 1 = a$.

	Finally, assumption \ref{ass:existence of antiparticles} states that $a^*$ is the inverse of $a$.
\end{proof}

We will often write $ab = a \times b$ for the product of elements $a, b \in I$.

%%%%%%%%%%%%%%%%%%%%%%%%%%%%%%%%%%%%%%%%%%%%%%%%%%%%%%%%%%%%%%%%%%%%%%%%%%%%%%%%%%%
%%%%%%%%%%%%%%%%%%%%%%%%%%%%%%%%%%%%%%%%%%%%%%%%%%%%%%%%%%%%%%%%%%%%%%%%%%%%%%%%%%%
\subsection{Braided tensor category} \label{subsec:braided tensor category}

It is well understood how to associate a braided tensor category to a pure state on a quantum spin system \cite{naaijkens2011localized, cha2020stability, ogata2022derivation}. In this section we recap this constuction in the very simple setting of assumptions \ref{ass:local string operators}-\ref{ass:locality of intertwiners}.

%%%%%%%%%%%%%%%%%%%%%%%%%%%%%%%%%%%%%%%%%%%%%%%%%%%%%%%%%%%%%%%%%%%%%%%%%%%%%%%%%%%
\subsubsection{Category of automorphisms}

Fix a unit vector $\hat f \in \R^2$, representing a `forbidden direction'. We say a cone $\Lambda_{a,\hat v,\theta}$ with axis $\hat v$ and opening angle $\theta$ is \emph{forbidden} if it contains the forbidden direction $\hat f$, \ie if $\hat v \cdot \hat f > \cos(\theta/2)$. A cone that is not forbidden is said to be \emph{allowed}. Let $\Delta$ be the group of automorphisms generated by $w_{a, \Lambda}$ and their inverses for $a \in I$ and $\Lambda$ allowed.

Since $\pi_1$ is a faithful representation of $\caA$ we will often identify $\caA$ with its image $\pi_1(\caA)$ and simply write $\rho$ instead of $\pi_1 \circ \rho$ for automorphisms $\rho$ on $\caA$.

The automorphisms in $\Delta$ are the objects of a $C^*$-category with morphisms 
\begin{equation}
	(\rho, \sigma) := \{  V \in \caB(\caH) \, : \, V \rho(O) = \sigma(O) V \,\,\, \text{for all } \,\,\, O \in \caA \}.
\end{equation}
Morphisms are referred to as intertwiners. (Direct sums of objects can be constructed as in Lemma 6.1 of \cite{naaijkens2011localized}).

\begin{lemma} \label{lem:elements of Delta are simple}
	Each $\rho \in \Delta$ is supported on some allowed cone, and there is a unique $a \in I$ such that $\pi_1 \circ \rho \simeq \pi_a$.
\end{lemma}

\begin{proof}
	Since $\Delta$ is generated by the $w_{a, \Lambda}$ and their inverses the first claim follows by noting that $w_{a, \Lambda}^{-1}$ is supported on $\Lambda$, and for any two allowed cones $\Lambda_1, \Lambda_2$ there is an allowed cone $\Lambda$ such that $\Lambda \supset \Lambda_1, \Lambda_2$, so compositions of automorphisms supported on allowed cones are also supported on allowed cones. For the second claim, note first that $\pi_1 \circ w_{a, \Lambda}^{-1} \simeq \pi_{a^*}$. Indeed, 
	$\pi_{a^*} \simeq \pi_1 \circ w_{a^*, \Lambda} 
	= \pi_1 \circ w_{a^*, \Lambda} \circ w_{a, \Lambda} \circ w_{a, \Lambda}^{-1} 
	\simeq \pi_{1} \circ w_{a, \Lambda}^{-1}$ 
	where we used assumption \ref{ass:abelian fusion}. Thus, $\rho$ is a finite composition of automorphisms $\rho_{1} \circ \cdots \circ \rho_n$ such that $\pi_1 \circ \rho_i \simeq \pi_{a_i}$ for some $a_i \in I$. It follows from assumption \ref{ass:abelian fusion} that $\pi_1 \circ \rho \simeq \pi_{a}$ with $a = a_1 \times \cdots \times a_n$.
\end{proof}

\begin{lemma} \label{lem:unique morphisms}
	The hom-sets $(\rho, \sigma)$ are either zero- or one-dimensional.
\end{lemma}

\begin{proof}
	By Lemma \ref{lem:elements of Delta are simple} there are unique $a, b \in I$ such that $\pi_a \simeq \pi_1 \circ \rho$ and $\pi_b \simeq \pi_1 \circ \sigma$. Each morphism in $(\rho, \sigma)$ then yields a distinct intertwiner from $\pi_a$ to $\pi_b$. In particular, the space of intertwiners from $\pi_a$ to $\pi_b$ has at least the dimension of $(\rho, \sigma)$. Since $\pi_a$ and $\pi_b$ are irreducible representations, this dimension is at most one.
\end{proof}

It follows from the above reasoning that every non-trivial hom-set contains a unitary intertwiner.
The need for a forbidden direction will become clear in the next section, where we equip $\Delta$ with a tensor product structure.

%%%%%%%%%%%%%%%%%%%%%%%%%%%%%%%%%%%%%%%%%%%%%%%%%%%%%%%%%%%%%%%%%%%%%%%%%%%%%%%%%%%
\subsubsection{Tensor product structure}

Let
\begin{equation}
	\caB := \overline{  \bigcup_{\text{allowed} \, \Lambda} \caR(\Lambda)  }^{\norm{\cdot}} \subset \caB(\caH).
\end{equation}

\begin{lemma}[Proposition 4.6 \cite{naaijkens2011localized}] \label{lem:extension to B}
	Each $\rho \in \Delta$ has a unique extension $\overline \rho$ to $\caB$ that is weakly continuous on $\caR(\Lambda)$ for any allowed cone $\Lambda$. Moreover, if $\rho$ is supported in an allowed cone $\Lambda$, then $\overline \rho( \caR(\Lambda) ) = \caR(\Lambda)$. In particular, $\overline \rho(\caB) = \caB$.
\end{lemma}

\begin{proof}
	By Lemma \ref{lem:elements of Delta are simple} the automorphism $\rho$ is supported on an allowed cone $\Lambda$, and $\pi_1 \circ \rho \simeq \pi_a$ for some anyon type $a \in I$. It then follows from assumption \ref{ass:local string operators} that there is a forbidden cone $\Lambda'$ with $\Lambda \cap \Lambda' = \emptyset$ such that $\pi_1 \circ \rho \simeq \pi_1 \circ w_{a, \Lambda'}$. Let $V \in \caB(\caH)$ be a unitary implementing this equivalence. We have for any $O \in \caA_{\Lambda}$ that
	\begin{equation}
		\rho(O) = V w_{a, \Lambda'}(O) V^* = V O V^*.
	\end{equation}
	We define the action of $\overline \rho$ on $\caR(\Lambda)$ by $\Ad(V)$, which is weakly continuous, and is uniquely determined by the action of $\rho$ on $\caA_{\Lambda}$. Clearly this action on $\caR(\Lambda)$ does not depend on the choice of $\Lambda'$. It follows that the extensions to $\caR(\Lambda)$ for different $\Lambda$ are consistent with each other. Thus the extension $\overline \rho$ is well-defined on all of $\caB$.

	Finally, $\rho(\caA_{\widetilde \Lambda}) = \caA_{\widetilde \Lambda}$ for any allowed cone $\widetilde \Lambda \supset \Lambda$ and weak continuity then implies $\overline \rho(\caR(\widetilde \Lambda)) = (\rho(\caA_{\widetilde \Lambda}))'' = \caR(\widetilde \Lambda)$ for all such cones.
\end{proof}

\begin{lemma} \label{lem:morphisms are local}
	Let $\Lambda$ be an allowed cone and denote by $\Delta_{\Lambda}$ the subgroup of $\Delta$ generated by $w_{a, \widetilde \Lambda}$ with $\widetilde \Lambda \subset \Lambda$ and their inverses. If $\rho, \sigma \in \Delta_{\Lambda}$ then $(\rho, \sigma) \subset \caR(\Lambda)$.	
\end{lemma}

\begin{proof}
	By assumption \ref{ass:locality of intertwiners} all intertwiners between $w_{a, \widetilde \Lambda}$ and $w_{a, \widetilde \Lambda'}$ belong to $\caR(\Lambda)$. Since the adjoint of such an intertwiner is an intertwiner from $w_{a, \widetilde \Lambda'}^{-1}$ to $w_{a, \widetilde \Lambda}^{-1}$, also all interwiners between inverses are in $\caR(\Lambda)$. It remains to show that intertwiners from $w_{a, \widetilde \Lambda}$ to $w_{a^*, \widetilde \Lambda'}^{-1}$ are in $\caR(\Lambda)$. To see this, note that assumption \ref{ass:locality of intertwiners} implies that any intertwiner $V$ from $w_{a, \widetilde \Lambda} \otimes w_{a^*, \widetilde \Lambda'}$ to $\id$ is an element of $\caR(\Lambda)$. But $V \otimes 1_{w_{a^*, \widetilde \Lambda'}^{-1}} = V$ is an intertwiner from $w_{a, \widetilde \Lambda}$ to $w_{a^*, \widetilde \Lambda'}^{-1}$. Since the space of intertwiners is one-dimensional, this shows that $(w_{a, \widetilde \Lambda}, w_{a^*, \widetilde \Lambda'}) \subset \caR(\Lambda)$ for any $\widetilde \Lambda, \widetilde \Lambda' \subset \Lambda$.

	Now take $\rho, \sigma \in \Delta_{\Lambda}$ and suppose there are $U, V \in \caR(\Lambda)$ such that $U \in (\rho, w_{a, \widetilde \Lambda})$ and $V \in (\sigma, w_{b, \widetilde \Lambda'})$ for some $a, b \in I$ and some $\widetilde \Lambda, \widetilde \Lambda' \subset \Lambda$. Then $U \otimes V \in \caR(\Lambda)$ is an intertwiner from $\rho \otimes \sigma$ to $w_{a, \widetilde \Lambda} \otimes w_{b, \widetilde \Lambda'}$. It follows from assumption \ref{ass:locality of intertwiners} that any intertwiner $W \in (w_{a, \widetilde \Lambda} \otimes w_{b, \widetilde \Lambda' }, w_{ab, \Lambda})$ belongs to $\caR(\Lambda)$, so $W(U \otimes V) \in (\rho \otimes \sigma, w_{ab, \Lambda})$ belongs to $\caR(\Lambda)$.

	Since the assumptions on $\rho$ and $\sigma$ hold for all generators of $\Delta_{\Lambda}$, it follows by induction that for all $\rho \in \Delta_{\Lambda}$ the intertwiners $(\rho, w_{a, \Lambda})$ belong to $\caR(\Lambda)$. Now take $U \in (\rho, w_{a, \Lambda})$ and $V \in (\sigma, w_{a, \Lambda})$, then $V^*U \in \caR(\Lambda)$ is an intertwiner from $\rho$ to $\sigma$, and since the $(\rho, \sigma)$ are at most one-dimensional, this proves the claim.
\end{proof}

We equip the $C^*$-category $\Delta$ with a tensor product structure as follows. For $\rho, \sigma \in \Delta$ we put
\begin{equation}
	\rho \otimes \sigma := \rho \circ \sigma,
\end{equation}
and for $R \in (\rho, \rho')$ and $S \in (\sigma, \sigma')$ we define
\begin{equation}
	R \otimes S := R \, \overline \rho( S ) \in (\rho \otimes \sigma, \rho' \otimes \sigma').
\end{equation}

%%%%%%%%%%%%%%%%%%%%%%%%%%%%%%%%%%%%%%%%%%%%%%%%%%%%%%%%%%%%%%%%%%%%%%%%%%%%%%%%%%%
\subsubsection{Braiding}

\begin{figure}
\centering
\includegraphics[width = 0.4\textwidth]{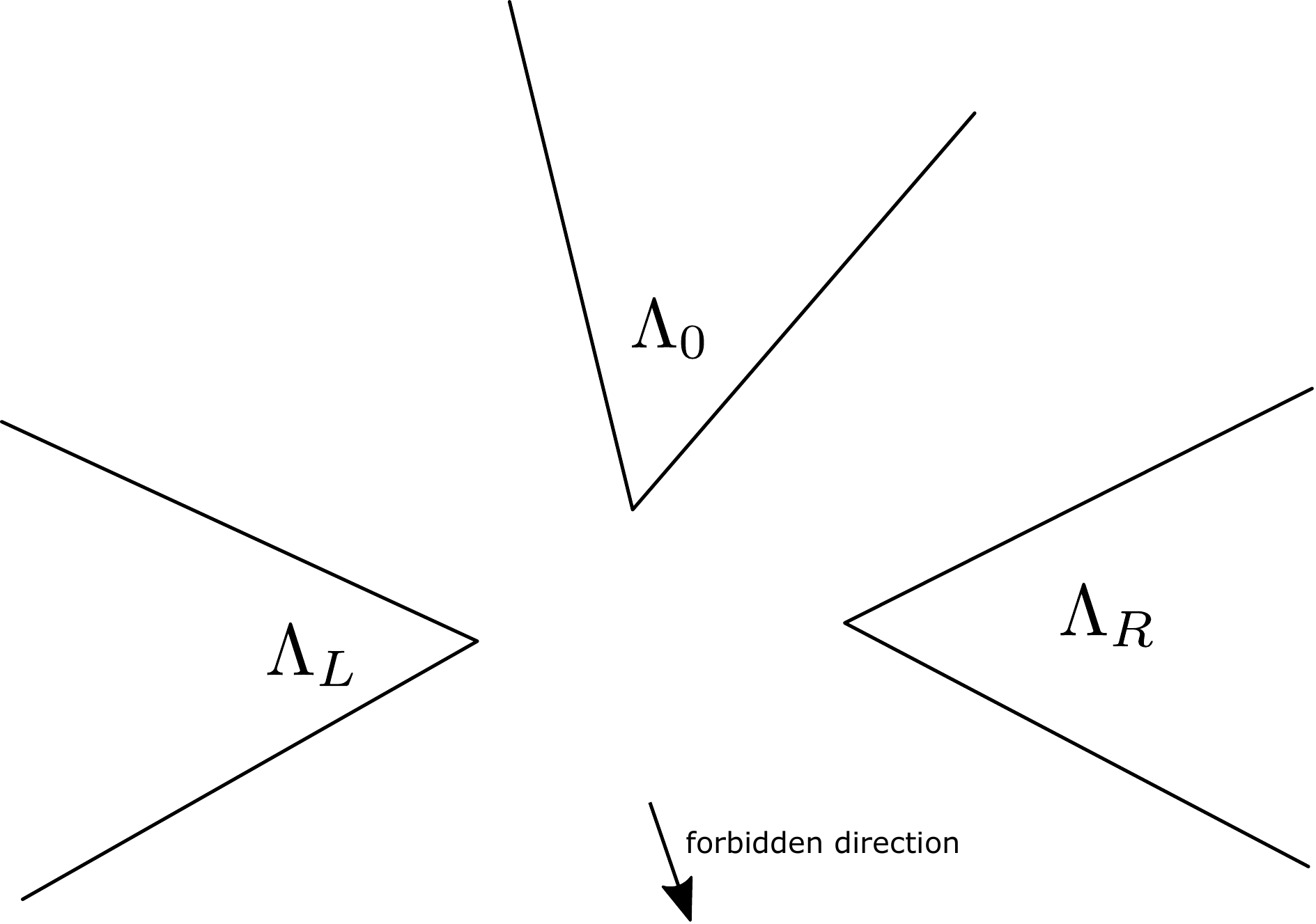}
\caption{Cones $\Lambda_0$, $\Lambda_L$ and $\Lambda_R$ used in the definition of the braiding intertwiners $\ep(\rho, \sigma)$.}
\label{fig:braiding setup}
\end{figure}

Consider two automorphisms $\rho, \sigma \in \Delta_{\Lambda_0}$ both supported in an allowed cone $\Lambda_0$. Pick allowed cones $\Lambda_L$ and $\Lambda_R$ as in Figure \ref{fig:braiding setup}. \ie the disjoint allowed cones $\Lambda_R, \Lambda_0$ and $\Lambda_L$ are arranged in a counterclockwise order from the forbidden direction, and there are allowed cones $\widetilde \Lambda_L \supset \Lambda_L \cup \Lambda_0$ and $\widetilde \Lambda_R \supset \Lambda_R \cup \Lambda_0$ such that $\widetilde \Lambda_L \cap \Lambda_R = \widetilde \Lambda_R \cap \Lambda_L  = \emptyset$. We sey $\Lambda_L$ is to the left of $\Lambda_0$, and $\Lambda_R$ is to the right of $\Lambda_0$. By Assumption \ref{ass:local string operators} there are automorphisms $\rho_L \in \Delta_{\Lambda_L}$ supported in $\Lambda_L$ and $\sigma_R \in \Delta_{\Lambda_R}$ supported in $\Lambda_R$ such that there are unitary $U \in (\rho, \rho_L)$ and $V \in (\sigma, \sigma_R)$.

\begin{definition} \label{def:braiding}
	The braiding intertwiner $\ep(\rho, \sigma) \in (\rho \otimes \sigma, \sigma \otimes \rho)$ is given by
	\begin{equation}
		\ep(\rho, \sigma) := (V^* \otimes U^*)(U \otimes V) = V^* \overline \rho(V).
	\end{equation}
\end{definition}
To get the last equality, we use $\overline \sigma_R(U) = U$. % by Assumption 4...
One verifies that $\epsilon(\rho,\sigma)$ is indeed an intertwiner by the fact that $\rho_L,\sigma_R$ commute.

\begin{lemma}\label{lem:braiding independent of intertwiners}
	The braiding $\ep(\rho, \sigma)$ is independent of the choice of $\rho_L$ and $\sigma_R$.
\end{lemma}

\begin{proof}
	Choose different $\rho'_L$ and $\sigma'_R$ with morphisms $U' \in (\rho, \rho'_L)$ and $V' \in (\sigma, \sigma'_R)$. Then % $U' U^* \in (\rho_L, \rho'_L) \subset \caR(\Lambda_L)$ and 
	$V'V^* \in (\sigma_R, \sigma'_R) \subset \caR(\Lambda_R)$ by Lemma \ref{lem:morphisms are local}. The new choice $\rho'_L, \sigma'_R$ leads to a braiding intertwiner
	\begin{align*}
		\ep'(\rho, \sigma) &= (V'^* \otimes U'^*)(U' \otimes V') = V'^* \overline \rho (V') \\
				   &= V^* (V' V^*)^* \overline \rho \big( (V' V^*) V  \big) = V^* \overline \rho( V ) \\
				   &= \ep(\rho, \sigma)
	\end{align*}
	where we used $\overline \rho( V'V^* ) = V'V^*$ since $V'V^* \in \caR(\Lambda_R)$ and $\rho$ is supported in $\Lambda_0$.
\end{proof}

\begin{lemma} \label{lem:braid equations}
	The braiding intertwiners satisfy the braid equations
	\begin{align*}
		\ep(\rho \otimes \sigma, \tau) &= \big( \ep(\rho, \tau) \otimes 1_{\sigma} \big) \big( 1_{\rho} \otimes \ep(\sigma, \tau) \big) \\
		\ep(\rho, \sigma \otimes \tau) &= \big( 1_{\sigma} \otimes \ep(\rho, \tau) \big) \big( \ep(\rho, \sigma) \otimes 1_{\tau}  \big)
	\end{align*}
	where $1_{\rho} = \I \in (\rho, \rho)$.
\end{lemma}

\begin{proof}
	Let us prove the first equation, the second is shown in the same way. Choose $\rho_L, \sigma_L$ supported in $\Lambda_L$ and morphisms $U_{\rho} \in (\rho, \rho_L)$, $U_{\sigma} \in (\sigma, \sigma_L)$. Choose $\tau_R$ supported in $\Lambda_R$ and a morphism $V_{\tau} \in (\tau, \tau_R)$. Then
	\begin{align*}
		\ep(\rho \otimes \sigma, \tau) &= \big( V_{\tau}^* \otimes (U_{\rho} \otimes U_{\sigma})^* \big) \big( (U_{\rho} \otimes U_{\sigma} ) \otimes V_{\tau} \big) \\
					       &= V_{\tau}^*  \overline{\rho \otimes \sigma}( V_{\tau} ) = V_{\tau}^* \overline \rho \big(  \overline \sigma \big( V_{\tau} \big) \big) \\
					       &= V_{\tau}^* \overline \rho(V_{\tau}) \overline \rho(V_{\tau}^*) \overline \rho \big( \overline \sigma( V_{\tau}) \big) \\
					       &= \big( \ep(\rho, \tau) \otimes 1_{\sigma} \big) \big( \overline \rho( V_{\tau}^* \overline \sigma(V_{\tau}) ) \big) \\
					       &= \big( \ep(\rho, \tau) \otimes 1_{\sigma} \big)\big( 1_{\rho} \otimes \ep(\sigma, \tau) \big).
	\end{align*}
\end{proof}

%%%%%%%%%%%%%%%%%%%%%%%%%%%%%%%%%%%%%%%%%%%%%%%%%%%%%%%%%%%%%%%%%%%%%%%%%%%%%%%%%%%
%%%%%%%%%%%%%%%%%%%%%%%%%%%%%%%%%%%%%%%%%%%%%%%%%%%%%%%%%%%%%%%%%%%%%%%%%%%%%%%%%%%
\subsection{Braided fusion category}

We construct a category whose objects are labeled by the anyon types $a \in I$. This category is obtained by a `skeletonization' of the category $\Delta$.

%%%%%%%%%%%%%%%%%%%%%%%%%%%%%%%%%%%%%%%%%%%%%%%%%%%%%%%%%%%%%%%%%%%%%%%%%%%%%%%%%%%
\subsubsection{Fusion and F-symbols}

Fix an allowed cone $\Lambda_0$ and write $w_{a} := w_{a, \Lambda_0}$. Pick unitary intertwiners $\Omega(a, b) \in (w_a \otimes w_b, w_{a \times b}) \subset \caR(\Lambda_0)$ called fusion operators. The $\Omega(a, b)$ are unique up to phase. Note that as an automorphism on $\caA$, we have $\Ad(\Omega(a, b)) = w_{a\times b} \otimes (w_a \otimes w_b)^{-1} \in \Delta_{\Lambda_0}$. The unitaries
\begin{align*}
	\Omega(ab, c) (\Omega(a, b) \otimes 1_c) &= \Omega(ab, c) \Omega(a, b) \\
	\Omega(a, bc) (1_a \otimes \Omega(b, c)) &= \Omega(a, bc) \overline w_a( \Omega(b, c) )
\end{align*}
are both intertwiners from $w_{abc}$ to $w_{a} \otimes w_b \otimes w_c$. Since $(w_{abc}, w_a \otimes w_b \otimes w_c)$ is one-dimensional, there are phases $F(a, b, c) \in U(1)$ such that
\begin{equation} \label{eq:F-symbols defined}
	\Omega(ab, c) \Omega(a, b) = F(a, b, c) \times \Omega(a, bc) \overline w_a(\Omega(b, c)).
\end{equation}
These $F(a, b, c)$ are the $F$-symbols. Figure \ref{fig:F-symbols} gives a graphical representation of Eq. \eqref{eq:F-symbols defined}.

\begin{figure}
\centering
\includegraphics[width = 0.4\textwidth]{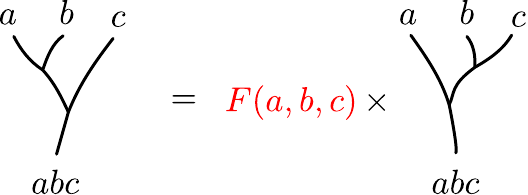}
\caption{Graphical representation of Eq. \eqref{eq:F-symbols defined}, defining the F-symbols $F(a, b, c)$. Each node represents a fusion operator. The diagrams represents two different compositions of fusion operators both yielding intertwiners from $w_a \otimes w_b \otimes w_c$ to $w_{abc}$.}
\label{fig:F-symbols}
\end{figure}

The F-symbols satisfy a pentagon equation, which in our setting of abelian anyons takes the form of a cocycle relation.
\begin{proposition} \label{prop:pentagon}
	The F-symbols satisfy
	\begin{equation}
		(\dd F)(a, b, c, d) := \frac{F(a, b, c) F(a, bc, d) F(b, c, d)}{F(ab, c, d) F(a, b, cd)} = 1.
	\end{equation}
\end{proposition}

\begin{proof}
	A graphical proof is shown in Figure \ref{fig:pentagon}. In equations, we have
	\begin{align*}
		\Omega(abc, d) \Omega(ab, c) \Omega(a, b) &= F(ab, c, d) \times \Omega(ab, cd) \overline w_{ab} \big( \Omega(c, d) \big) \Omega(a, b) \\
							  &= F(ab, c, d) \times \Omega(ab, cd) \Omega(a, b) \overline w_a \big( \overline w_b \big( \Omega(c, d) \big)  \big) \\
							  &= F(ab, c, d) F(a, b, cd) \times \Omega(a, bcd) \overline w_{a}( \Omega(b, cd)) \overline w_a( \overline w_b( \Omega(c, d)) ) \\
							  &= F(ab, c, d) F(a, b, cd) \times \Omega(a, bcd) \overline w_a \big( \Omega(b, cd) \overline w_b( \Omega(c, d) ) \big)
	\end{align*}
	but also
	\begin{align*}
		\Omega(abc, d) \Omega(ab, c) \Omega(a, b) &= F(a, b, c) \times \Omega(abc, d) \Omega(a, bc) \overline w_a( \Omega(b, c) ) \\
							  &= F(a, b, c) F(a, bc, d) \times \Omega(a, bcd) \overline w_{a}( \Omega(bc, d) ) \overline w_s( \Omega(b, c) ) \\
							  &= F(a, b, c) F(a, bc, d) \times \Omega(a, bcd) \overline w_a \big( \Omega(bc, d) \Omega(b, c) \big) \\
							  &= F(a, b, c) F(a, bc, d) F(b, c, d) \times \Omega(a, bcd) \overline w_a \big( \Omega(b, cd) \overline w_b( \Omega(c, d) )  \big).
	\end{align*}
	And the desired equality follows.
\end{proof}

\begin{figure}
\centering
\includegraphics[width = 0.6\textwidth]{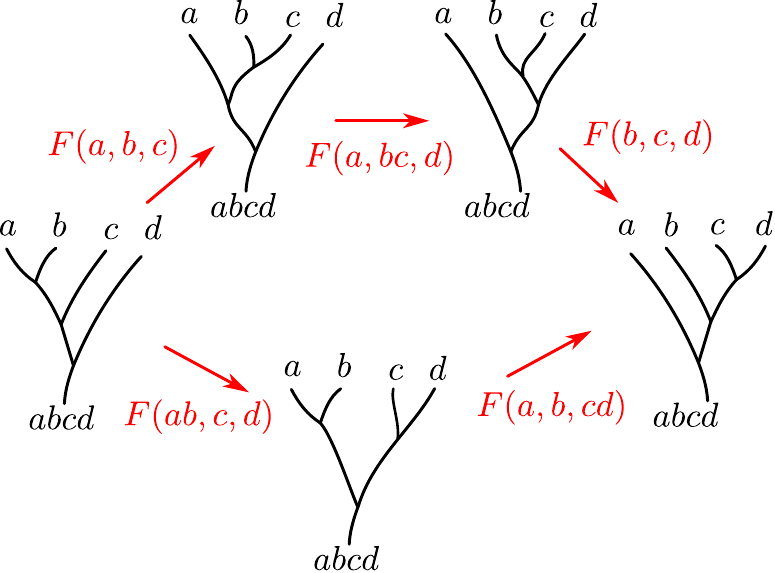}
\caption{A graphical proof of the Pentagon equation.}
\label{fig:pentagon}
\end{figure}

%%%%%%%%%%%%%%%%%%%%%%%%%%%%%%%%%%%%%%%%%%%%%%%%%%%%%%%%%%%%%%%%%%%%%%%%%%%%%%%%%%%
\subsubsection{Braiding and R-symbols}

We simply set $\ep(a, b) := \ep(w_a, w_b)$ for any $a, b \in I$. The unitaries $\Omega(b, a) \ep(a, b)$ and $\Omega(a, b)$ are both intertwiners from $w_a \otimes w_b$ to $w_{ab}$. Since $(w_a \otimes w_b, w_{ab})$ is one-dimensional, there exist phases $R(a, b) \in U(1)$ such that
\begin{equation} \label{eq:R-symbols defined}
	\Omega(b, a) \ep(a, b) = R(a, b) \times \Omega(a, b).
\end{equation}
The phases $R(a, b)$ are the R-symbols. Figure \ref{fig:R-symbols} gives a graphical representation of Eq. \eqref{eq:R-symbols defined}.

\begin{figure}
\centering
\includegraphics[width = 0.4\textwidth]{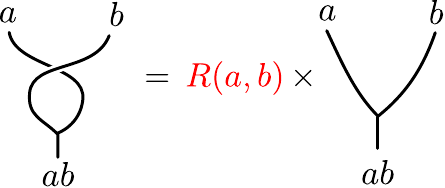}
\caption{Graphical representation of Eq. \eqref{eq:R-symbols defined}, defining the R-symbols $R(a, b)$. The point where the $a$-line passes under the $b$-line represents the braiding intertwiner $\ep(a, b)$.}
\label{fig:R-symbols}
\end{figure}

%%%%%%%%%%%%%%%%%%%%%%%%%%%%%%%%%%%%%%%%%%%%%%%%%%%%%%%%%%%%%%%%%%%%%%%%%%%%%%%%%%%
\subsubsection{Yang-Baxter equation}

The braidings $\ep(a, b)$ and fusions $\Omega(a, b)$ satisfy the Yang-Baxter equations, see Figure \ref{fig:Yang-Baxter}.
\begin{proposition} \label{prop:Yang-Baxter}
	We have
	\begin{equation}
		\overline w_c \big(\Omega(a, b) \big)  \ep(a, c) \overline w_a \big( \ep(b, c) \big) = \ep(ab, c)  \Omega(a, b)
	\end{equation}
	and
	\begin{equation}
		\Omega(b, c) \overline w_b \big( \ep(a, c) \big) \ep(a, b) = \ep(a, bc) \overline w_a \big( \Omega(b, c) \big).
	\end{equation}
	
\end{proposition}

\begin{proof}
	By the braid equations (Proposition \ref{lem:braid equations}) and the fact that $\Ad(\Omega(a, b)) \in \Delta$, we have
	\begin{align*}
		\ep(ab, c) &= \ep \big( \Ad(\Omega(a, b)) \otimes (w_a \otimes w_b), w_c\big)  \\
			   &= \big( \ep(\Ad(\Omega(a, b)), w_c) \otimes 1_{w_a \otimes w_b}  \big) \big( 1_{\Ad(\Omega(a, b))} \otimes \ep( w_a \otimes w_b, w_c ) \big) \\
			   &= \ep( w_c, \Ad(\Omega(a, b)) )^* \, \Omega(a, b) \big( \ep(w_a \otimes w_b, w_c) \big) \Omega(a, b)^* \\
			   &= \overline w_c \big( \Omega(a, b) \big)  (\ep(a, c) \otimes 1_b)(1_a \otimes \ep(b, c)) \Omega(a, b)^* \\
			   &= \overline w_c \big( \Omega(a, b) \big) \ep(a, c) \overline w_a \big( \ep(b, c) \big) \Omega(a, b)^*
	\end{align*}
	where we used $\ep(\rho, \sigma) = \ep(\sigma, \rho)^*$ in obtaining the third line. The fourth line uses that $\ep( w_c, \Ad(\Omega(a, b)) )=\Omega(a, b)\overline w_c \big( \Omega(a, b)^* \big)$ by Lemma \ref{lem:braiding independent of intertwiners} since $\Omega(a, b)^*\in (\Ad(\Omega(a, b)),\text{id})$.
	This proves the first Yang-Baxter equation. The second is shown similarly.	
\end{proof}

\begin{figure}
\centering
\includegraphics[width = 0.6\textwidth]{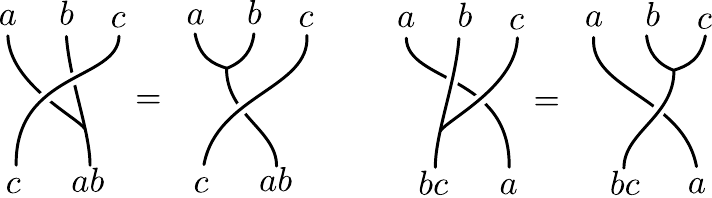}
\caption{Graphical representation of the Yang-Baxter equations.}
\label{fig:Yang-Baxter}
\end{figure}

%%%%%%%%%%%%%%%%%%%%%%%%%%%%%%%%%%%%%%%%%%%%%%%%%%%%%%%%%%%%%%%%%%%%%%%%%%%%%%%%%%%
\subsubsection{Hexagon equation}

\begin{figure}
\centering
\includegraphics[width = 1.0\textwidth]{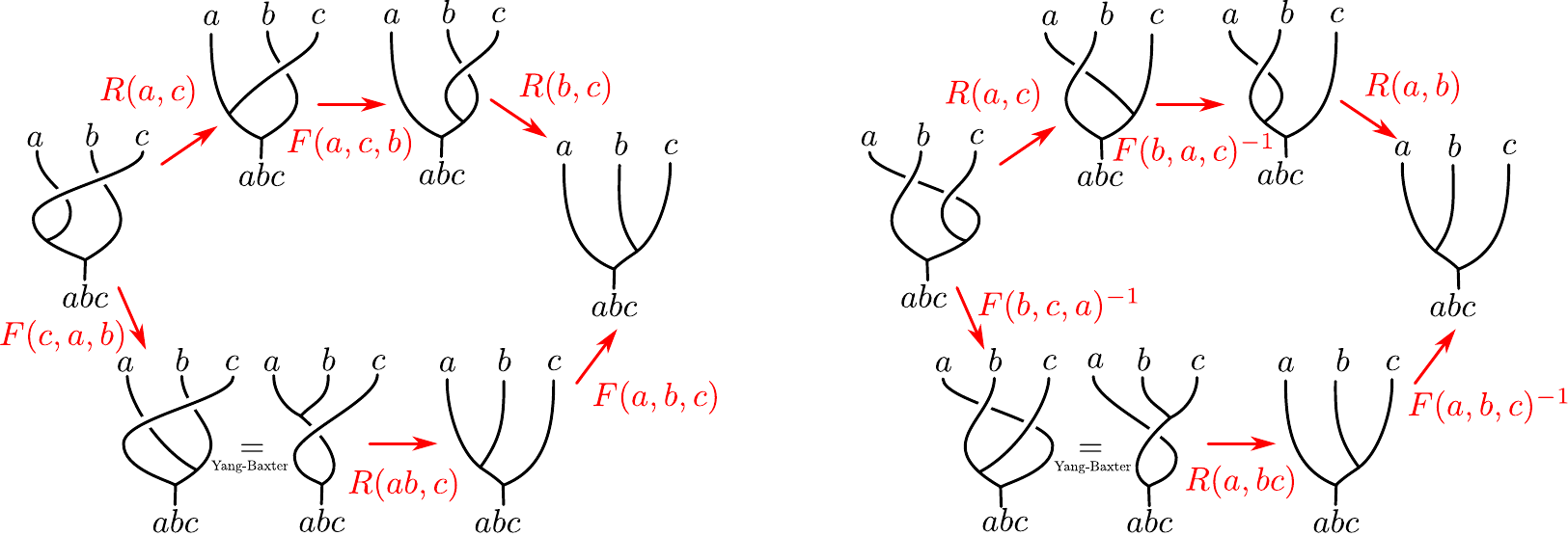}
\caption{Graphical representations of the first and second hexagon equations.}
\label{fig:hexagon}
\end{figure}

Using the Yang-Baxter equation, we obtain the Hexagon equation, see Figure \ref{fig:hexagon}.
\begin{proposition} \label{prop:hexagon}
	The F and R-symbols satisfy the hexagon equations
	\begin{equation} \label{eq:first hexagon}
		\frac{F(a, b, c) F(c, a, b)}{F(a, c, b)} = \frac{R(a, c) R(b, c)}{R(ab, c)}
	\end{equation}
	and 
	\begin{equation} \label{eq:second hexagon}
		\frac{F(a, b, c) F(b, c, a)}{F(b, a, c)} = \frac{R(a, bc)}{R(a, b) R(a, c)}.
	\end{equation}
	
\end{proposition}

\begin{proof}
	The left diagram in Figure \ref{fig:hexagon} suggests the following two equalities of morphisms:
	\begin{align*}
		\Omega(ca, b) \Omega(c, a) \ep(a, c) \overline w_a(\ep(b, c)) &= R(a, c) \times \Omega(ca, b) \Omega(a, c) \overline w_a(\ep(b, c)) \\
									      &= R(a, c) F(a, c, b) \times \Omega(a, cb) \overline w_a (\Omega(c, b)) \overline w_a( \ep(b, c)) \\
									      &= R(a, c) F(a, c, b) \times \Omega(a, cb) \overline w_a \big( \Omega(c, b) \ep(b, c) \big) \\
									      &= R(a, c) F(a, c, b) R(b, c) \times \Omega(a, cb) \overline w_a( \Omega(b, c) )
	\end{align*}
	and
	\begin{align*}
		\Omega(ca, b) \Omega(c, a) \ep(a, c) \overline w_a(\ep(b, c)) &= F(c, a, b) \times \Omega(c, ab) \overline w_c( \Omega(a, b) ) \ep(a, c) \overline w_a(\ep(b, c)) \\
									      &= F(c, a, b) \times \Omega(c, ab) \ep(ab, c) \Omega(a, b) \\
									      &= F(c, a, b) R(ab, c) \times \Omega(ab, c) \Omega(a, b) \\
									      &= F(c, a, b) R(ab, c) F(a, b, c) \times \Omega(a, bc) \overline w_a(\Omega(b, c))
	\end{align*}
	where we used the Yang-Baxter equation to obtain the second line. The coefficients of the right hand sides must be equal, yielding the first hexagon equation.

	The second hexagon equation is obtained in exactly the same way, following the right diagram in Figure \ref{fig:hexagon}.
\end{proof}

%%%%%%%%%%%%%%%%%%%%%%%%%%%%%%%%%%%%%%%%%%%%%%%%%%%%%%%%%%%%%%%%%%%%%%%%%%%%%%%%%%%
\subsubsection{Braided fusion category}

We now consider the category with objects $a \in I$ thought of as one-dimensional vector spaces over $\C$. The monoidal structure is given by the fusion rules:
\begin{equation}
	a \otimes b = a \times b.
\end{equation}
The unit is $1 \in I$ and the left and right unitors are simply the identity map. The associators $\al_{a, b, c} : (a \otimes b) \otimes c \rightarrow a \otimes (b \otimes c)$ are given by multiplication with $F(a, b, c)$. Proposition \ref{prop:pentagon} shows that these associators satisfy the pentagon identity, and since $F(a, 1, c) = 1$ for all $a, c$, also the triangle identity is satisfied.

The braiding $\ep_{a, b} : a \otimes b \rightarrow b \otimes a$ is given by multiplication with $R(a, b)$. Proposition \ref{prop:hexagon} shows that this braiding and the associators given by the F-symbols satisfy the hexagon identities.

All this extends uniquely to direct sums.

Finally, each anyon $a \in I$ has dual $a^*$.

%%%%%%%%%%%%%%%%%%%%%%%%%%%%%%%%%%%%%%%%%%%%%%%%%%%%%%%%%%%%%%%%%%%%%%%%%%%%%%%%%%%
\subsubsection{Dependence of F and R-symbols on the choice of $\Lambda_0$ and the phases of $\Omega(a, b)$.}

Suppose we chose different phases for the intertwiners $\Omega(a, b)$, \ie we consider
\begin{equation}
	\Omega'(a, b) = \chi(a, b) \Omega(a, b)
\end{equation}
for phases $\chi(a, b)$. This yields new F-symbols by
\begin{equation}
	\Omega'(ab, c) \Omega'(a, b) = F'(a, b, c) \times \Omega'(a, bc) \overline w_a (\Omega'(b, c))
\end{equation}
which are related to the original F-symbols by
\begin{equation} \label{eq:F transformation}
	F'(a, b, c) = (\dd \chi)(a, b, c) F(a, b, c) = \frac{\chi(b, c) \chi(a, bc)}{\chi(ab, c) \chi(a, b)}  F(a, b, c).
\end{equation}
\ie $F'$ is related to $F$ by the coboundary $\dd \chi$. It follows that only the cohomology class $[F] \in H^3(I, U(1))$ is well-defined.

The R-symbols are also affected by the different choice of phases. Indeed, the new R-symbols defined by
\begin{equation}
	\Omega'(b, a) \ep(a, b) = R'(a, b) \times \Omega'(a, b)
\end{equation}
are related to the old by
\begin{equation} \label{eq:R transformation}
	R'(a, b) = \frac{\chi(b, a)}{\chi(a, b)} R(a, b).
\end{equation}
It follows that the self-statistcs $R(a, a)$ and the double braidings $R(a, b) R(b, a)$ are invariants.

Next, we investigate the dependence of the F and R-symbols on the choice of allowed cone $\Lambda_0$. We will find no additional ambiguity beyond the one just discussed.

Let $\Lambda_0'$ be another allowed cone. Then there is an allowed cone $\widetilde \Lambda_0$ containing $\Lambda_0 \cup \Lambda_0'$. Denote $w'_a = w_{a, \Lambda_0'}$. Then there are unitaries $W_a \in \caR(\widetilde \Lambda_0)$ such that
\begin{equation} \label{eq:relation of new to old w_a}
	w_a' \circ w_a^{-1} = \Ad[W_a].
\end{equation}
These unitaries are unique up to phase.

This leads to new intertwiners $\Omega'(a, b)$, uniquely defined up to phase by
\begin{equation}
	\Ad[ \Omega'(a, b) ] = w'_{ab} \circ (w'_a \circ w'_b)^{-1}.
\end{equation}
A short computation shows that the right hand side becomes
\begin{equation}
	w'_{ab} \circ (w'_a \circ w'_b)^{-1} = \Ad \left[  W_{ab}  \Omega(a, b) \bar w_a(W_b^*)  W_a^* \right].
\end{equation}
It follows that we can choose phases for the $\Omega'(a, b)$ such that
\begin{equation} \label{eq:transformed Omega}
	\Omega'(a, b) = W_{ab} \Omega(a, b) \bar w_a(W_b^*) W_a^*.
\end{equation}
(Note that we already understood the dependence of the F and R-symbols on the choice of phases of the $\Omega'(a, b)$, so in this discussion we can choose these phases as it suits us.)

The new F-symbols are determined by
\begin{equation}
	\Omega'(ab, c) \Omega'(a, b) = F'(a, b, c) \times \Omega'(a, bc) \bar w'_a( \Omega'(b, c) ).
\end{equation}
Using Eq. \ref{eq:transformed Omega} we compute
\begin{equation}
	\Omega'(ab, c) \Omega'(a, b) = W_{abc} \Omega(ab, c) \Omega(a, b) \bar w_a \big( \bar w_b( W_c^* ) W_b^*  \big) W_a^*
\end{equation}
and
\begin{equation}
	\Omega'(a, bc) \bar w'_a( \Omega'(b, c) ) = W_{abc} \Omega(a, bc) \bar w_a( \Omega(b, c) )  \bar w_a \big( \bar w_b( W_c^* ) W_b^*  \big) W_a^*.
\end{equation}
It follows that $F'(a, b, c) = F(a, b, c)$ for all $a, b, c \in I$.

Recall that the braiding intertwiners $\ep(a, b)$ are defined in terms of automorphisms $w_a$ supported in $\Lambda_0$, $w_a^L$ supported in $\Lambda_L$ and $w_a^R$ supported in $\Lambda_R$ as follows. Pick intertwiners (unique up to phase) $U_a \in (w_a, w_a^L)$ and $V_a \in (w_a, w_a^R)$, then
\begin{equation}
	\ep(a, b) = (V_b^* \otimes U_a^*)(U_a \otimes V_b) = V_b^* \bar w_a( V_b ).
\end{equation}
It was shown in Lemma \ref{lem:braiding independent of intertwiners} that this braiding intertwiner is independent of the choice of the automorphisms $w_a^L$ and $w_a^R$ and therefore of the precise choice of cones $\Lambda_L$ and $\Lambda_R$. Moreover, $\ep(a, b)$ is independent of the choice of phase for the intertwiners $U_a$ and $V_a$.

Let's choose the left and right cones $\Lambda_L$ and $\Lambda_R$ such that they are to the left and right of $\widetilde \Lambda_0$ respectively.

With the new choice of automorphisms $w'_a$ related to the old $w_a$ by Eq. \eqref{eq:relation of new to old w_a} we get new intertwiners $U'_a = U_a W^*_a \in (w'_a, w_a^L)$ and $V'_a = V_a W^*_a \in (w'_a, w_a^R)$ and therefore new braiding intertwiners
\begin{equation}
	\ep'(a, b) = ((V'_b)^* \otimes (U'_a)^*)(U'_a \otimes V'_b) = (V'_b)^* \bar w'_a \big( V'_b \big).
\end{equation}
A short computation relates this to the braiding $\ep(a, b)$ as
\begin{equation} \label{eq:transformed braiding}
	\ep'(a, b) = W_b V_b^* W_a V_b \ep(a, b) \bar w_a( W^*_b ) W^*_a.
\end{equation}

The new $R$-symbol is determined by
\begin{equation}
	\Omega'(b, a) \ep'(a, b) = R'(a, b) \times \Omega'(a, b).
\end{equation}
Using Eqs. \eqref{eq:transformed Omega} and \eqref{eq:transformed braiding} the left hand side becomes
\begin{equation}
	\Omega'(b, a) \ep'(a, b) = W_{ba} \Omega(b, a) \ep(a, b) \bar w_a( W^*_b ) W^*_a
\end{equation}
and
\begin{equation}
	\Omega'(a, b) = W_{ab} \Omega(a, b) \bar w_a( W^*_b ) W^*_a
\end{equation}
so, noting that $ab = ba$, we find $R'(a, b) = R(a, b)$.

We conclude that Eqs. \eqref{eq:F transformation} and \eqref{eq:R transformation} are the only ambiguities in the F and R-symbols.

The braided fusion category with the old F and R-symbols is braided monoidally equivalent to the braided fusion category with the new F and R-symbols through an equivalence $F$ which is the identity map on objects $I$, and has natural tranformations $\Phi_{a, b} : F(a) \otimes F(b) \rightarrow F(a \otimes b)$ given by multiplication with the phase $\chi(a, b)$.

%%%%%%%%%%%%%%%%%%%%%%%%%%%%%%%%%%%%%%%%%%%%%%%%%%%%%%%%%%%%%%%%%%%%%%%%%%%%%%%%%%%
%%%%%%%%%%%%%%%%%%%%%%%%%%%%%%%%%%%%%%%%%%%%%%%%%%%%%%%%%%%%%%%%%%%%%%%%%%%%%%%%%%%
%%%%%%%%%%%%%%%%%%%%%%%%%%%%%%%%%%%%%%%%%%%%%%%%%%%%%%%%%%%%%%%%%%%%%%%%%%%%%%%%%%%
\section{The double semion state} \label{sec:double semion}

We construct an infinite volume version of the ground state of the double semion model, first introduced in \cite{levin2005string}. We identify superselection sectors corresponding to semion, anti-semion, and bound state anyons and find that the braided fusion category describing these anyons corresponds to the representation category of a twisted quantum double algebra $\caD^{\phi}(\Z_2)$.

%%%%%%%%%%%%%%%%%%%%%%%%%%%%%%%%%%%%%%%%%%%%%%%%%%%%%%%%%%%%%%%%%%%%%%%%%%%%%%%%%%%
%%%%%%%%%%%%%%%%%%%%%%%%%%%%%%%%%%%%%%%%%%%%%%%%%%%%%%%%%%%%%%%%%%%%%%%%%%%%%%%%%%%
\subsection{Construction of the double semion state} \label{sec:state construction}

\begin{figure}
\centering
\includegraphics[width = 0.4\textwidth]{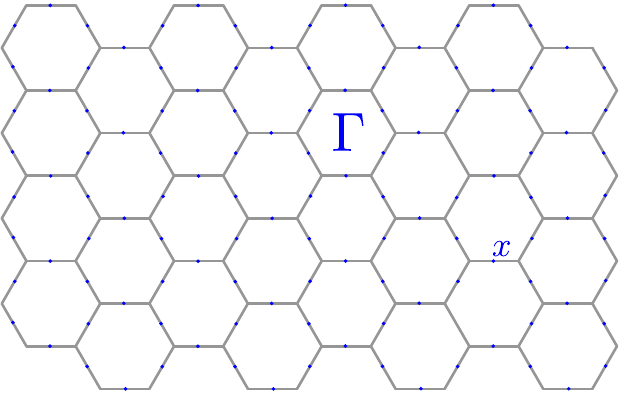}
\caption{The degrees of freedom of the double semion state live on the edges of a hexagon lattice.}
\label{fig:hexagonal lattice}
\end{figure}

\begin{figure}
\centering
\includegraphics[width = 0.4\textwidth]{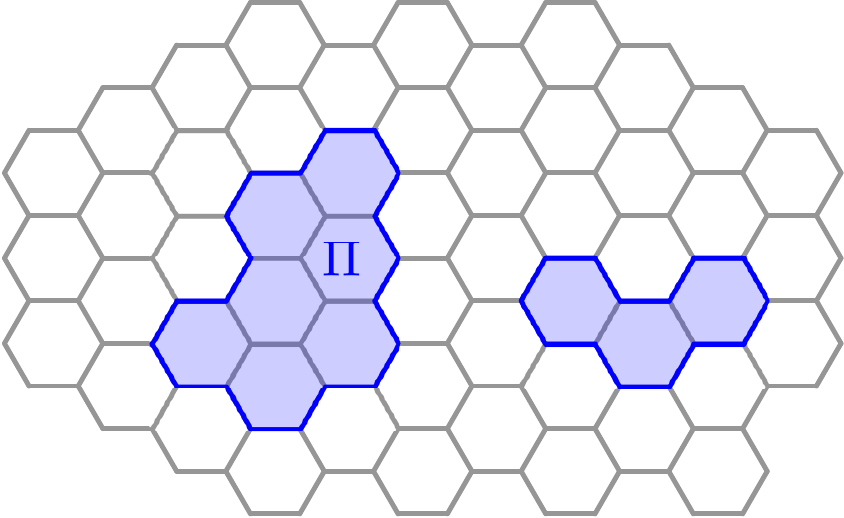}
\caption{Acting with $A_{\Pi}$ on the state $\omega_0$ yields a string configuration with two connected components.}
\label{fig:finite string}
\end{figure}

\begin{figure}
\centering
\includegraphics[width = 0.4\textwidth]{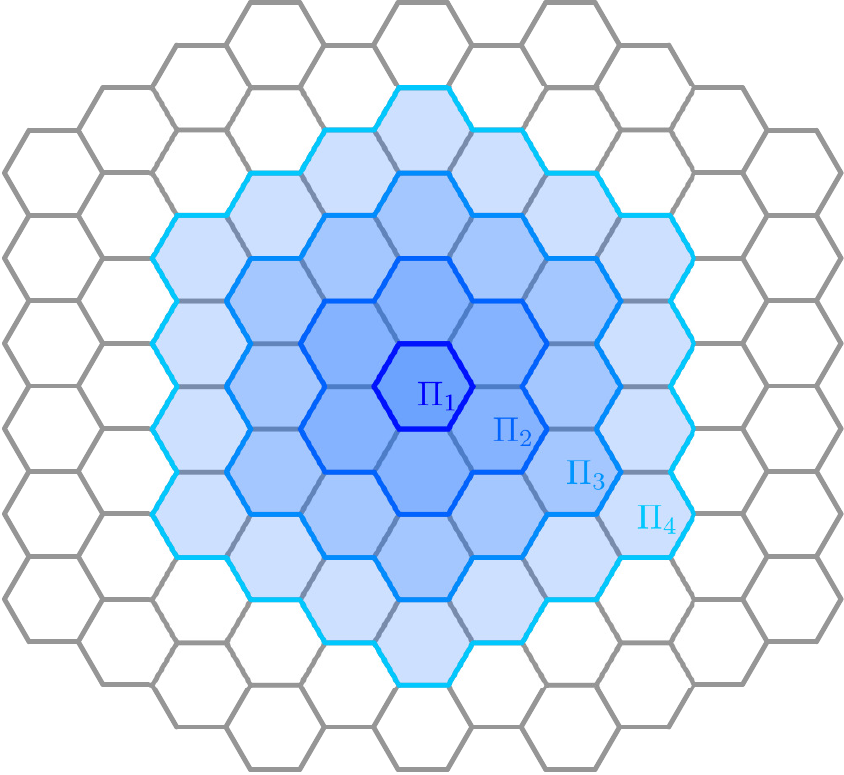}
\caption{}
\label{fig:increasing sequence}
\end{figure}

We take $\Gamma$ to be the (midpoints of the) edges of the hexagonal lattice (Figure \ref{fig:hexagonal lattice}) and $\caA_x \simeq \End(\C^2)$. We fix Pauli matrices $\sigma_x^X, \sigma_x^Y, \sigma_x^Z$ in each $\caA_x$.

We interpret $\sigma_x^Z = -1$ as the edge $x$ being occupied by a string, while $\sigma_x^Z=1$ means that the edge is unoccupied. Let $\omega_0$ be the pure product state without any strings, \ie $\omega_0(\sigma_x^Z) = 1$ for all $x \in \Gamma$.

for any hexagon $p$ let
\begin{equation}
	A_p = \prod_{x \in P} \sigma_x^X
\end{equation}
and for any finite set $\Pi$ of hexagons, let
\begin{equation}
	A_{\Pi} = \prod_{p \in \Pi} A_p=\prod_{p \in \partial\Pi} \sigma_x^X.
\end{equation}
Note that $A_{\Pi}$ produces a string around the region $\Pi$ when it acts on $\omega_0$, see Figure \ref{fig:finite string}.

Let $(\pi_0, \caH_0, \Omega_0 )$ be the GNS triple for $\omega_0$. Let $\Pi_n$ be an increasing sequence of sets of hexagons as in Figure \ref{fig:increasing sequence} and let
\begin{equation}
	\Omega_n := \sqrt{\frac{1}{2^{\abs{\Pi_n}}}}  \, \sum_{\Pi \subset \Pi_n} \, (-1)^{\sharp\Pi} \, A_{\Pi} \Omega_0 
\end{equation}
where $\sharp\Pi$ is the number of connected components of $\Pi$. \ie $\Omega_n$ is a superposition of closed string configurations supported in $\Pi_n$, with phases determined by the parity of the number of components of the string configuration.

The vectors $\Omega_n$ determine a sequence of pure states $\omega_n$ on $\caA$. The following theorem is proven in the appendix (Appendix \ref{app:purity})

\begin{theorem} \label{thm:purity}
	The sequence $\omega_n$ converges in the weak-* topology to a pure state $\omega$.
\end{theorem}

This pure state $\omega$ is the double semion state.

Let $(\pi_1, \caH, \Omega)$ be the GNS triple for $\omega$.

%%%%%%%%%%%%%%%%%%%%%%%%%%%%%%%%%%%%%%%%%%%%%%%%%%%%%%%%%%%%%%%%%%%%%%%%%%%%%%%%%%%
%%%%%%%%%%%%%%%%%%%%%%%%%%%%%%%%%%%%%%%%%%%%%%%%%%%%%%%%%%%%%%%%%%%%%%%%%%%%%%%%%%%
\subsection{String operators}

Let $P$ be an oriented edge-self-avoiding path in $\Gamma$. Following \cite{levin2005string}, we define three types of non-trivial string operators. The \emph{semion string} is given by
\begin{equation}
	W_S[P] := \left( \prod_{j \in P}  \sigma_j^X \right)  \left( \prod_{\mathrm{R-legs} \, j} \iu^{\frac{1 - \sigma_j^Z}{2}}  \right) \left( \prod_{\mathrm{L-vertices} \, v}  (-1)^{s_v}  \right)
\end{equation}
where $s_v = \frac{1}{4}(1 - \sigma_j^Z)(1 + \sigma_k^Z)$ and $j, k$ are the edges of $P$ that go in and out of the vertex $v$ respectively. The terms `R-legs' and `L-vertices' are explained in Figure \ref{fig:R-legs and L-vertices}.

The \emph{anti-semion} string is given by
\begin{equation}
	W_{\bar S}[P] := \left( \prod_{j \in P}  \sigma_j^X \right)  \left( \prod_{\mathrm{R-legs} \, j} (-\iu)^{\frac{1 - \sigma_j^Z}{2}}  \right) \left( \prod_{\mathrm{L-vertices} \, v}  (-1)^{s_v}  \right)
\end{equation}
and the \emph{bound-state string} is given by
\begin{equation}
	W_{B}[P] = \left( \prod_{\mathrm{R-legs} \, j} \sigma_j^Z \right).
\end{equation}

\begin{figure}
\centering
\includegraphics[width = 0.4\textwidth]{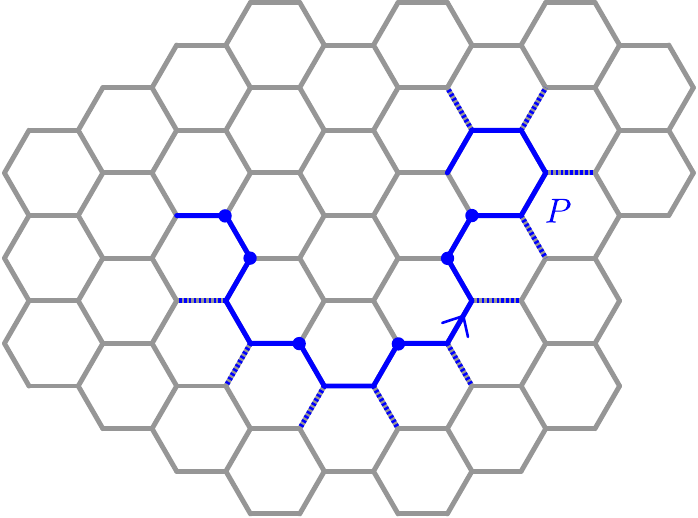}
\caption{An oriented path $P$ in solid blue with its L-vertices fattened and its R-legs marked with dotted lines.}
\label{fig:R-legs and L-vertices}
\end{figure}

We further define string operators for the vacuum sector $W_1[P] = \I$, all equal to the identity.
We set $I = \{1, S,\bar S, B\}$ and we denote by $w_a[P]$ the automorphism defined by conjugation with the (possibly formal) unitary $W_a[P]$.

These string operators have the following important property

\begin{proposition} \label{prop:closed strings leave the ground state invariant}
	Closed string operators leave the ground state invariant. \ie if $P$ is a closed string then
	\begin{equation}
		\omega \circ w_a[P] = \omega
	\end{equation}
	for all $a \in I$.
\end{proposition}

The proof is in appendix \ref{app:strings}.

%%%%%%%%%%%%%%%%%%%%%%%%%%%%%%%%%%%%%%%%%%%%%%%%%%%%%%%%%%%%%%%%%%%%%%%%%%%%%%%%%%%
\subsubsection{Definition of $w_{a, \Lambda}$} \label{ssubsec:choice of transportable string operators}

For any set $S \subset \R^2$, let $\Pi_{S}$ be the set of hexagons (regarded as open subsets of $\R^2$) that have some overlap with the set $S$. For a cone $\Lambda$, the boundary $\partial \Pi_{\Lambda}$ is an infinite oriented closed path, see Figure \ref{fig:boundary path}.

\begin{figure}
\centering
\includegraphics[width = 0.8\textwidth]{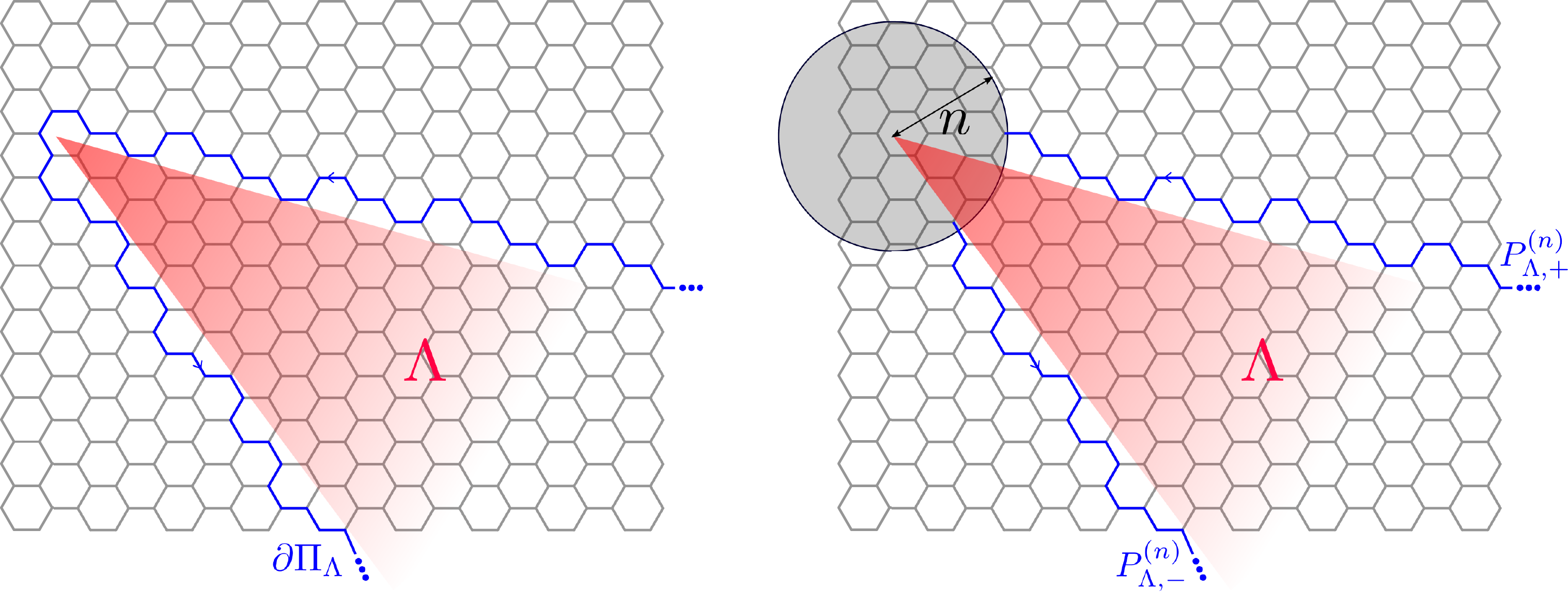}
\caption{A cone $\Lambda$ with the oriented path $\partial \Pi_{\Lambda}$ going around it in a counter clockwise direction.}
\label{fig:boundary path}
\end{figure}

For any cone $\Lambda$ whose opening angle is less than $\pi$, the edges of $\partial \Pi_{\Lambda}$ whose center lies a distance further than $n > 2$ from the apex of $\Lambda$ form two half-infinite paths $P^{(n)}_{\Lambda, +}$ and $P^{(n)}_{\Lambda, -}$ following the left and right legs of $\Lambda$, respectively.

Let $\Lambda$ be a cone, and let $\Lambda^{(L)}$ and $\Lambda^{(R)}$ be its left and right half-cones, see Figure \ref{fig:cone path}. Take $n > 2$ sufficiently large such that $w_{a, \Lambda} := w_a[P^{(n)}_{\Lambda^{(R)}, +}]$ and $v_{a, \Lambda} : = w_a[ P^{(n)}_{\Lambda^{(L)}, -} ]$ are supported in $\Lambda$ for all $a \in I$. Denote $P_{\Lambda} := P^{(n)}_{\Lambda^{(R)}, +}$ and $\overline P_{\Lambda} := P^{(n)}_{\Lambda^{(L)}, -}$.

\begin{figure}
\centering
\includegraphics[width = 0.4\textwidth]{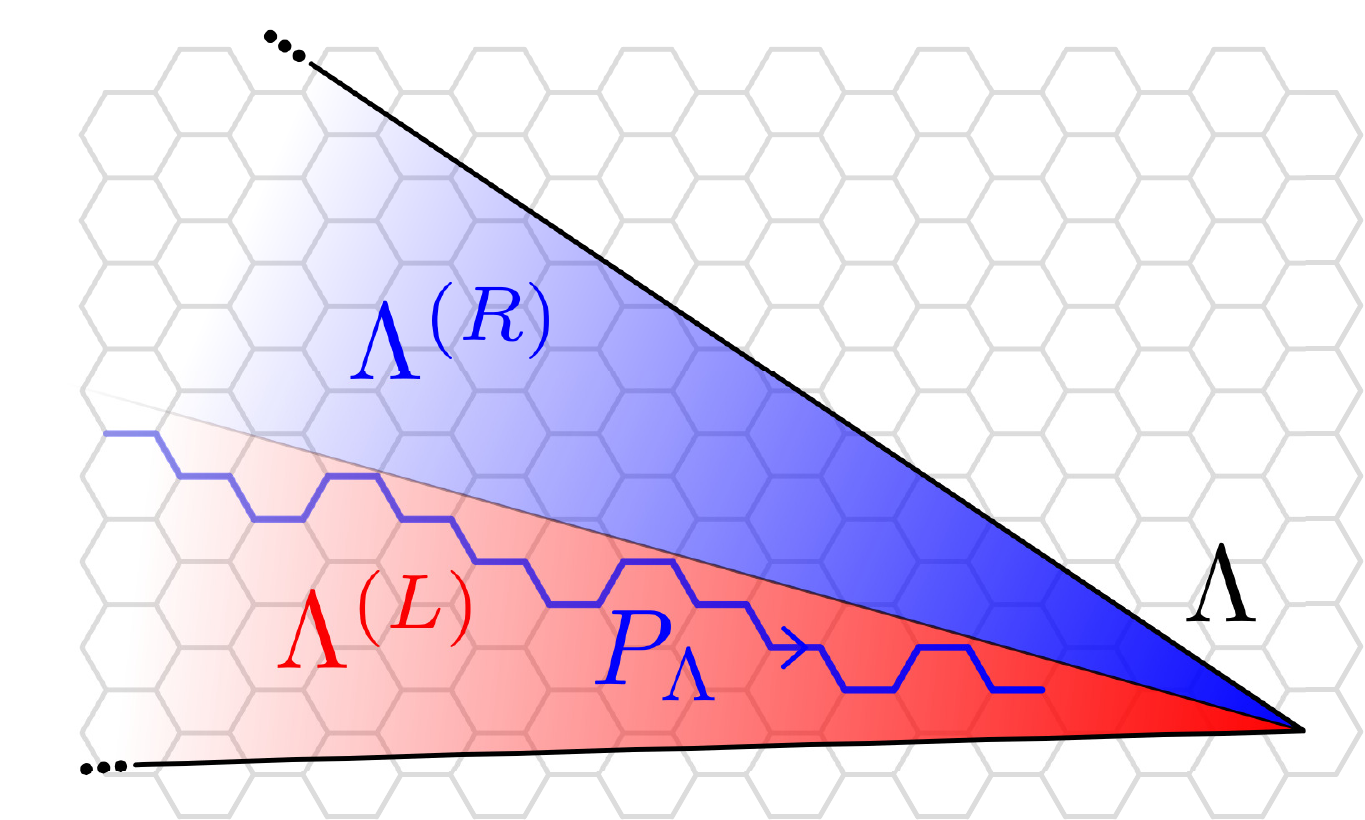}
\caption{A cone $\Lambda$ divided into its left and right cones $\Lambda^{(L)}$ and $\Lambda^{(R)}$. The path $P_{\Lambda}$ is the largest part of $\partial \Pi_{\Lambda^{(R)}}$ such that $w_a[P_{\Lambda}]$ is supported in $\Lambda$.}
\label{fig:cone path}
\end{figure}

%%%%%%%%%%%%%%%%%%%%%%%%%%%%%%%%%%%%%%%%%%%%%%%%%%%%%%%%%%%%%%%%%%%%%%%%%%%%%%%%%%%
\subsubsection{Fusion rules}

Consider semion-semion fusion. For any oriented path $P$, the automorphism $w_S[P] \circ w_S[P]$ is given by conjugation with the formal unitary
\begin{equation}
	\Omega_{S, S}[P] := \left( \prod_{\mathrm{R-legs}}  \sigma_j^Z \right) \left( \prod_{\substack{\mathrm{L-verts} \\ v = (i, j)}} \sigma_i^Z \sigma_j^Z \right).
\end{equation}

Let $i$ and $f$ be the initial and final edges of the path $P$ (if they exist) and let $v[P]$ be the automorphism given by conjugation with
\begin{equation} \label{eq:V string defined}
	V[P] := \Omega_{S, S}[P] \times \sigma_i^Z \sigma_f^Z.
\end{equation}

\begin{lemma} \label{lem:v leaves the ground state invariant}
	We have $\omega \circ v[P] = \omega$ for any path $P$.
\end{lemma}

\begin{proof}
	We first take $P$ finite and show that $v[P]$ leaves any $\omega_n$ invariant. Recall that $\omega_n$ is the expectation value of the GNS vector
	\begin{equation}
		\Omega_n = \sqrt{\frac{1}{2^{\abs{\Pi_n}}}}\, \sum_{\Pi \subset \Pi_n} (-1)^{\sharp \Pi} \, A_{\Pi} \Omega_0.
	\end{equation}
	We have
	\begin{equation}
		V[P] A_{\Pi} \Omega_0 = A_{\Pi} \Omega_0.
	\end{equation}
	Indeed, $A_{\Pi}$ is a product of $\sigma_j^X$ for all $j \in \partial \Pi$, a closed path. Now, any such closed path supports an even number of factors $\sigma_j^Z$ of the unitary $V[P]$. Indeed, if $\partial \Pi$ travels along $P$, then the two edges of $P$ along an R-leg carry no $\sigma^Z$, while the two edges along an L-vertex both have a $\sigma^Z$. The closed path $\partial \Pi$ must enter/leave the path $P$ an even number of times. If it enters through an R-leg, it picks up a $\sigma^Z$ from the R-leg. If it enters through an L-vertex, then it picks up exactly one of the $\sigma^Z$'s of the two edges of $P$ next to the L-vertex. Finally, if $\partial \Pi$ enters $P$ through an endpoint of $P$, then the factors $\sigma_i^Z$, $\sigma_f^Z$ at the initial/final edges ensure that a factor $\sigma^Z$ is picked up. In all, we see that $V[P] A_{\Pi} = A_{\Pi} V[P]$, because the computation involves an even number of commutations of a $\sigma^X$ with a $\sigma^Z$. Obviously $V[P] \Omega_0 = \Omega_0$ so $V[P] A_{\Pi} \Omega_0 = A_{\Pi} \Omega_0$ and $V[P] \Omega_n = \Omega_n$. It follows that $\omega_n \circ v[P] = \omega_n$ for any $n$ and any finite $P$.

	It follows immediately that $\omega \circ v[P] = \omega$ for any finite $P$. If $P$ is infinite, then for any strictly local observalbe $O$ we can find a finite $P'$ such that $v[P](O) = v[P'](O)$ so $(\omega \circ v[P])(O) = (\omega \circ v[P'])(O) = \omega(O)$. Since the strictly local observables are dense in $\caA$, this proves the claim.
\end{proof}

\begin{lemma} \label{lem:v is in the cone}
	If $v[P]$ is supported in a cone $\Lambda$ then $\pi_1 \circ v[P] \simeq \pi_1$, and the unitary implementing this equivalence belongs to the von Neumann algebra $\caR(\Lambda)$.
\end{lemma}

\begin{proof}
	The unitary equivalence $\pi_1 \circ v[P] \simeq \pi_1$ follows immediately from the previous Lemma. Let $V$ be the unitary implementing this equivalence, \ie
	\begin{equation}
		V O V^* =  v[P](O)
	\end{equation}
	for all $O \in \caA$, and $V \Omega = \Omega$. (We identify $\caA$ with its image under the faithful representation $\pi_1$)

	If $P$ where finite, then actually $V \in \caA_{\Lambda} \subset \caR(\Lambda)$. If $P$ is infinite, let $P_n$ be the path consisting of edges of $P$ whose midpoints lie in $\Pi_n$. Then $V[P_n] \in \caA_{\Lambda}$ has $V[P_n] \Omega = \Omega$ for all $n$, and for any strictly local obervables $O, O'$ we have
	\begin{equation}
		\langle O \Omega, V O' \Omega \rangle = \langle O \Omega, v[P](O') V \Omega \rangle = \lim_{n \uparrow \infty} \langle O \Omega, v[P_n](O') \Omega \rangle = \lim_{n \uparrow \infty} \langle O \Omega, V[P_n] O' \Omega \rangle.
	\end{equation}
	Since the vectors $O \Omega$, $O' \Omega$ for $O, O'$ strictly local observables are dense in $\caH$, this shows that the sequence $V[P_n]$ converges weakly to $V$. Since $V[P_n] \in \caA_{\Lambda}$ for all $n$, it follows that $V \in \caR(\Lambda)$.	
\end{proof}

\begin{lemma} \label{lem:S x S = 1}
	For any cone $\Lambda$ we have that $\pi_1 \circ w_{S, \Lambda} \circ w_{S, \Lambda} \simeq \pi_1$, and the unitary $V_{\Lambda}$ implementing this equivalence belongs to the von Neumann algebra $\caR(\Lambda)$.
\end{lemma}

\begin{proof}
	By definition, $w_{S, \Lambda} = w_S[ P_{\Lambda}]$ so $w_{S, \Lambda} \circ w_{S, \Lambda} = \Ad(\sigma_f^Z) \circ v[P_{\Lambda}]$ where $f$ is the final edge of the half-infinite path $P_{\Lambda}$. From Lemma \ref{lem:v is in the cone}, we find that there exists a unitary $V_{\Lambda} \in \caR(\Lambda)$ such that $v[P_{\Lambda}] = \Ad(V_{\Lambda})$, hence
	\begin{equation}
		\pi_1 \circ w_{S, \Lambda} \circ w_{S, \Lambda} = \pi_1 \circ \Ad(\sigma_f^ZV_{\Lambda}),
	\end{equation}
	proving the claim.
\end{proof}

We can now easily show
\begin{proposition} \label{prop:double semion fusion rules}
	For each cone $\Lambda$ there are unitaries $\Omega(a, b) \in \caR(\Lambda)$ such that
	\begin{equation}
		\Ad(\Omega(a, b)) \circ w_{a, \Lambda} \circ w_{b, \Lambda} = w_{a \times b, \Lambda}
	\end{equation}
	for all $a, b \in I = \{1, S, \bar S, B\}$ and where $\times$ is an abelian product on $I$ given by
	\begin{center}
		\begin{tabular}{ |c|c c c c | } 
			\hline
			$\times$ & $1$ & $S$ & $\bar S$ & $B$ \\
			\hline
			$1$ & $1$ & $S$ & $\bar S$ & $B$ \\
			$S$ & $S$ & $1$ & $B$ & $\bar S$ \\
			$\bar S$ & $\bar S$ & $B$ & $1$ & $S$ \\
			$B$ & $B$ & $\bar S$ & $S$ & $1$ \\
			\hline
		\end{tabular}
	\end{center}
\end{proposition}

\begin{proof}
	Lemma \ref{lem:S x S = 1} shows that the claim holds for $S \times S = 1$. The rest of the claim follows from this case and
	\begin{equation}
		w_{\bar S, \Lambda} = w_{S, \Lambda} \circ w_{B,\Lambda} = w_{B, \Lambda} \circ w_{S, \Lambda}, \quad w_{B, \Lambda} \circ w_{B, \Lambda} = \id.
	\end{equation}
\end{proof}

%%%%%%%%%%%%%%%%%%%%%%%%%%%%%%%%%%%%%%%%%%%%%%%%%%%%%%%%%%%%%%%%%%%%%%%%%%%%%%%%%%%
\subsubsection{Transportability}

\begin{lemma} \label{lem:semion transport lemma}
	If $\Lambda_1$ and $\Lambda_2$ are cones with axes $\hat w_1$ and $\hat w_2$, both contained in a cone $\Lambda$ and such that $\hat w_1$ points to the right of $\hat w_2$ relative to $\Lambda$ (see Figure \ref{fig:transport construction}). Then $\pi_1 \circ w_{a, \Lambda_1} \simeq \pi_1 \circ v_{a, \Lambda_2}^{-1}$ and the unitary implementing this equivalence belongs to the von Neumann algebra $\caR(\Lambda)$.
\end{lemma}

\begin{proof}
	Fix points $x_1, x_2$ on the central axes of the cones $\Lambda_1, \Lambda_2$ such that the region $S$ bounded by the half-infinite parts of these axes starting at $x_1, x_2$, and the line between $x_1$ and $x_2$ is convex and has $w_{a, \partial \Pi_{S}}$ supported in $\Lambda$, see Figure \ref{fig:transport construction}.

	By construction, $w_{a, \partial \Pi_{S}}$ differs from $w_{a, \Lambda_1} \circ v_{a, \Lambda_2}$ by the action of a local unitary $W$ supported on $\Lambda$. Since $\partial \Pi_{S}$ is a closed path, it follows from Proposition \ref{prop:closed strings leave the ground state invariant} that there exists a unitary $V \in \caB(\caH)$ such that $\pi_1 \circ w_{a, \partial \Pi_{S}} = \Ad(V) \circ \pi_1$ and $B \Omega = \Omega$, hence
	\begin{equation}
		\pi_1 \circ w_{a, \Lambda} = \Ad(VW) \circ \pi_1 \circ v_{a, \Lambda_2}^{-1}.
	\end{equation}
	This shows the required unitary equivalence. It remains to show that $V \in \caR(\Lambda)$.

	Let $S_n = S \cap B_b$ where $B_n$ is the disk of radius $n$ centered at the origin of $\R^2$. Then $\partial \Pi_{S_n}$ are closed paths and the automorphisms $w_{\partial \Pi_{S_n}} = \Ad( W_a[\partial \Pi_{S_n}] )$ leave the ground state invariant, and are supported in $\Lambda$. In particular, there exist phases $\phi_n$ such that $V_n := \phi_n W_a[\partial \Pi_{S_n}]$ satisfies $V_n \Omega = \Omega$. For any strictly local observables $O, O' \in \pi_1(\caA)$ we have
	\begin{equation}
		\langle O \Omega, V O' \Omega \rangle = \langle O \Omega, w_{a, \partial \Pi_S}(O') V \Omega \rangle = \lim_{n \uparrow \infty} \langle O \Omega, w_{a, \partial \Pi_{S_n}}(O') V_n \Omega \rangle = \lim_{n \uparrow \infty} \langle O \Omega, V_n O' \Omega \rangle.
	\end{equation}
	Since the vectors $O \Omega, O' \Omega$ are dense in $\caH$, this shows that $V_n$ converges weakly to $V$. since each $V_n$ is in $\caA_{\Lambda}$, we conclude that $V \in \caR(\Lambda)$.
\end{proof}

\begin{figure}
\centering
\includegraphics[width = 0.4\textwidth]{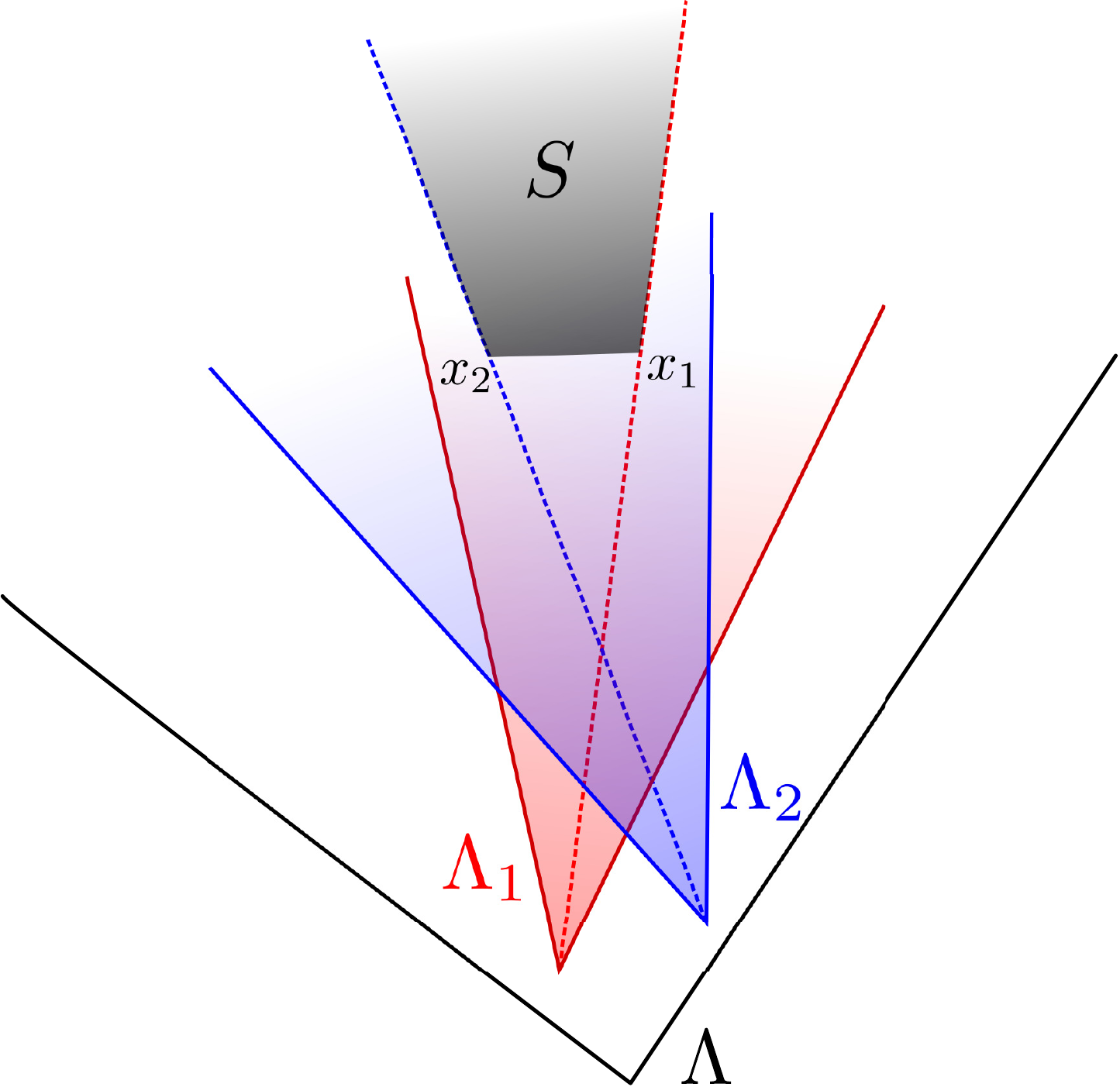}
\caption{The axis of $\Lambda_1$ points to the right of the axis of $\Lambda_2$ relative to the cone $\Lambda$. The region $S$ has $w_a[\partial \Pi_S]$ supported in $\Lambda$. Moreover, the path $\partial \Pi_{S}$ differs from the union of the paths $P_{\Lambda_1}$, $\overline P_{\Lambda_2}$ by a finite number of edges.}
\label{fig:transport construction}
\end{figure}

\begin{proposition} \label{prop:transportability}
	If $\Lambda_1$ and $\Lambda_2$ are cones both contained in a cone $\Lambda$, then $\pi_1 \circ w_{a, \Lambda_1} \simeq \pi_1 \circ w_{a, \Lambda_2}$ and the unitary implementing this equivalence belongs to the von Neumann algebra $\caR(\Lambda)$.
\end{proposition}

\begin{proof}
	Let $\hat w_1, \hat w_2$ be the axes of the cones $\Lambda_1, \Lambda_2$ and take a cone $\Lambda_3 \subset \Lambda$ such that its axis $\hat w_3$ points to the right of both $\hat w_1$ amd $\hat w_2$ relative to $\Lambda$. Then Lemma \ref{lem:semion transport lemma} implies that there unitaries $V_1, V_2 \in \caR(\Lambda)$ such that
	\begin{equation}
		\pi_1 \circ w_{a, \Lambda_1} = \Ad(V_1) \circ \pi_1 \circ v_{a, \Lambda_3}^{-1}, \quad \pi_1 \circ w_{a, \Lambda_2} = \Ad(V_2) \circ \pi_1 \circ v_{a, \Lambda_3}^{-1},
	\end{equation}
	hence
	\begin{equation}
		\pi_1 \circ w_{a, \Lambda_1} = \Ad(V_2^* V_1) \circ \pi_1 \circ w_{a, \Lambda_2}.
	\end{equation}
	Since $V_2^* V_1 \in \caR(\Lambda)$, this proves the claim.
\end{proof}

%%%%%%%%%%%%%%%%%%%%%%%%%%%%%%%%%%%%%%%%%%%%%%%%%%%%%%%%%%%%%%%%%%%%%%%%%%%%%%%%%%%
\subsubsection{Distinct sectors}

Fix a cone $\Lambda_0$ with axis $(0, 1)$ and let $\pi_a := \pi_1 \circ w_{a, \Lambda_0}$ for $a \in I$.

\begin{proposition} \label{prop:four distinct anyon types}
	For all $a, b \in I$, we have $\pi_a \simeq \pi_b$ if and only if $a = b$.
\end{proposition}

\begin{proof}
	For any $n$ large enough such that the endpoint of $P_{\Lambda_0}$ is contained in $\Pi_{n-2}$, consider the S-matrix 
	\begin{equation}
		S_{ab} := \frac{1}{2} (\omega \circ w_{a, \Lambda_0})( W_b[\partial \Pi_n]   ).
	\end{equation}

	An easy calculation shows that these quantities are independent of $n$, and given by
	\begin{equation}
		S = \frac{1}{2} \begin{bmatrix} 1 & 1 & 1 & 1 \\
			1 & -1 & 1 & -1 \\
			1 & 1 & -1 & -1 \\
			1 & -1 & -1 & 1
		\end{bmatrix}.
	\end{equation}
	
	It follows that for any $a \neq b$ there is a $c$ such that $(\omega \circ w_{a, \Lambda_0})(W_c[\partial \Pi_n]) = - (\omega \circ w_{b, \Lambda_0})(W_c[\partial \Pi_n])$ for all $n$ sufficiently large. Corollary 2.6.11 of \cite{bratteli2012operator} then implies that $\pi_a$ and $\pi_b$ are not unitarily equivalent.
\end{proof}

%%%%%%%%%%%%%%%%%%%%%%%%%%%%%%%%%%%%%%%%%%%%%%%%%%%%%%%%%%%%%%%%%%%%%%%%%%%%%%%%%%%
\subsubsection{Verification of assumptions}

The four faithful irreducible representations $\pi_1, \pi_S, \pi_{\bar S}, \pi_B$ defined by $\pi_a = \pi_1 \circ w_{a, \Lambda_0}$ for $a \in \{1, S, \bar S, B\} = I$ are pairwise not unitarily equivalent by Proposition \ref{prop:four distinct anyon types}.

For any cone $\Lambda$ and any $a \in I$ we defined an automorphism $w_{a, \Lambda}$ supported in $\Lambda$. This collection of automorphisms satisfies assumption \ref{ass:local string operators} by Proposition \ref{prop:transportability}. Assumptions \ref{ass:abelian fusion} and $\ref{ass:existence of antiparticles}$ are verified by Proposition \ref{prop:double semion fusion rules}. Finally, assumption \ref{ass:locality of intertwiners} holds by Propositions \ref{prop:transportability} and \ref{prop:double semion fusion rules}.

%%%%%%%%%%%%%%%%%%%%%%%%%%%%%%%%%%%%%%%%%%%%%%%%%%%%%%%%%%%%%%%%%%%%%%%%%%%%%%%%%%%
%%%%%%%%%%%%%%%%%%%%%%%%%%%%%%%%%%%%%%%%%%%%%%%%%%%%%%%%%%%%%%%%%%%%%%%%%%%%%%%%%%%
\subsection{Computation of F-symbols}

Having fixed the cone $\Lambda_0$ with axis $(0, 1)$, we write $w_a := w_{a, \Lambda_0}$ for $a \in I$.

Let $f$ be the endpoint of the path $P_{\Lambda_0}$ and let $V \in \caR(\Lambda_0)$ be the unitary such that $\pi_1 \circ v[P_{\Lambda_0}] = \Ad(V) \circ \pi_1$ and $V \Omega = \Omega$ provided by Lemma \ref{lem:v leaves the ground state invariant}. The proof of Lemma \ref{lem:S x S = 1} shows that
\begin{equation}
	\Ad(\Omega(S, S) ) \circ (w_S \otimes w_S) = w_1
\end{equation}
with $\Omega(S, S) = \sigma_f^Z V$.

Using
\begin{equation}
	w_{\bar S} = w_B \circ w_S = w_S \circ w_{B}, \quad w_B \circ w_B = \id
\end{equation}
we find that
\begin{equation}
	\Ad(\Omega(a, b)) \circ (w_a \otimes w_b) = w_{a \times b}
\end{equation}
for all $a, b \in I$ with fusion interwiners $\Omega(a, b)$ given in Table \ref{tab:fusion intertwiners}.
\begin{table} 
	\begin{center} 
	\begin{tabular}{ |c|c c c c | } 
		\hline
		$\Omega(a, b)$ & $1$ & $s$ & $\bar s$ & $b$ \\
		\hline
		$1$ & $\I$ & $\I$ & $\I$ & $\I$ \\
		$s$ & $\I$ & $V \sigma_f^Z$ & $V \sigma_f^Z$ & $\I$ \\
		$\bar s$ & $\I$ & $V \sigma_f^Z$ & $V \sigma_f^Z$ & $\I$ \\
		$b$ & $\I$ & $\I$ & $\I$ & $\I$ \\
		\hline
	\end{tabular}
	\end{center}
	\caption{The fusion intertwiners $\Omega(a, b)$ for the double semion state.}
	\label{tab:fusion intertwiners}
\end{table}

In order to compute the F-symbols, we first show
\begin{lemma} \label{lem:transformation of non-trivial fusion intertwiner}
	\begin{equation}
		\overline w_S( V \sigma_f^Z ) = - V \sigma_f^Z, \quad \overline w_B(V \sigma_f^Z) = V \sigma_f^Z, \quad \overline w_{\bar S}(V \sigma_f^Z) = - V \sigma_f^Z.
	\end{equation}
\end{lemma}

\begin{proof}
	Since $f$ is the last edge of the path $P_{\Lambda_0}$, we have $w_S(\sigma_f^Z) = w_{\bar S}(\sigma_f^Z) = - \sigma_f^Z$ and $w_B(\sigma_f^Z) = \sigma_f^Z$. It remains to show that $\overline w_{S}(V) = \overline w_{\bar S}(V) = \overline w_B(V) = V$.

	Since $V$ is the weak limit of the sequence $V_n = V[P_n]$ where $P_n$ is the path consisting of edges of $P_{\Lambda_0}$ whose midpoints lie in $\Pi_n$ (cf. the proof of Lemma \ref{lem:v is in the cone}), it is sufficient to show $w_{S}(V_n) = w_{\bar S}(V_n) = w_B(V_n) = V_n$. This follows similarly to the argument in the proof of Lemma \ref{lem:v is in the cone}. Since $V_n$ is a product of $\sigma_Z^s$ we have that $w_B(V_n) = V_n$, and
	\begin{equation}
		w_S(V_n) = w_{\bar S}(V_n) = \left( \prod_{j \in P_n} \sigma_j^X \right)  V_n \left(  \prod_{j \in P_n} \sigma_j^X \right).
	\end{equation}
	By design, the unitary $V_n$ has an even number of $\sigma^Z$'s on the path $P_n$. Indeed, there are two factors of $\sigma^Z$ for every L-vertex, zero for every R-leg, and another two for the enpoints. We conclude that $w_S(V_n) = W_{\bar S}(V_n) = V_n$ for all $n$.
\end{proof}

We can now start computing the $F$-symbols. If in Eq. \eqref{eq:F-symbols defined} we take $a = 1$ then
\begin{equation}
	 \Omega(b, c) \Omega(1, b) = F(1, b, c)\Omega(1, bc) \Omega(B, C).
\end{equation}
Since $\Omega(1, b) = \Omega(1, bc) = \I$ we find that $F(1, b, c) = 1$ for all $b, c$.

Similarly we find $F(a, 1, c) = F(a, b, 1) = 1$ for all $a, b, c$.

Let us now consider $F$-symbols that involve the bound state $B$, for example
\begin{equation}
	\Omega(Bb, c) \Omega(B, b) = F(B, b, c) \Omega(B, bc) \overline w_B( \Omega(b, c) ).
\end{equation}
Since $\Omega(B, b) = \Omega(B, bc) = \I$ this reduces to
\begin{equation}
	\Omega(Bb, c) = F(B, b, c) \overline w_B( \Omega(b, c) ).
\end{equation}

If $b = B$ or $c = B$ then then $\Omega(Bb, b) = \Omega(b, c) = \I$ so $F(B , b, c) = 1$. If $b, c \in \{S, \bar S\}$ then $\Omega(Bb, c) = \Omega(b, c) = V \sigma_f^Z$, so using Lemma \ref{lem:transformation of non-trivial fusion intertwiner} we find again $F(B, b, c) = 1$ for all $b, c$.

Similar considerations show that $F(B, b, c) = F(a, B, c) = F(a, b, B) = 1$ for all $a, b, c$.

Finally, we consider the case where $a, b, c \in \{ S, \bar S \}$. Then since $ab, bc \in \{ 1, B \}$ we have $\Omega(ab, c) = \Omega(a, bc) = \I$ so
\begin{equation}
	V \sigma_f^Z  = F(a, b, c) \overline w_{a}( V \sigma_f^Z ).
\end{equation}
Using Lemma \ref{lem:transformation of non-trivial fusion intertwiner} we conclude that $F(a, b, c) = -1$ for $a, b, c \in \{S, \bar S\}$.

%%%%%%%%%%%%%%%%%%%%%%%%%%%%%%%%%%%%%%%%%%%%%%%%%%%%%%%%%%%%%%%%%%%%%%%%%%%%%%%%%%%
%%%%%%%%%%%%%%%%%%%%%%%%%%%%%%%%%%%%%%%%%%%%%%%%%%%%%%%%%%%%%%%%%%%%%%%%%%%%%%%%%%%
\subsection{Computation of R-symbols}

Choose cones $\Lambda_L$ with axis $(-1, 0)$ and $\Lambda_R$ with axis $(1, 0)$, both disjoint from $\Lambda_0$ as in Figure \ref{fig:semion braiding setup}.

To compute the braiding intertwiners $\ep(a, b) = \ep(w_a, w_b)$, set $w_a^L = v_a[\Lambda_L]^{-1}$ and $w_a^R = v_a[\Lambda_R]^{-1}$ for all $a = 1, S, \bar S, B$. (Recall that the automorphisms $v_{a, \Lambda}$ were defined in section \ref{ssubsec:choice of transportable string operators}).

\begin{figure}
\centering
\includegraphics[width = 0.6\textwidth]{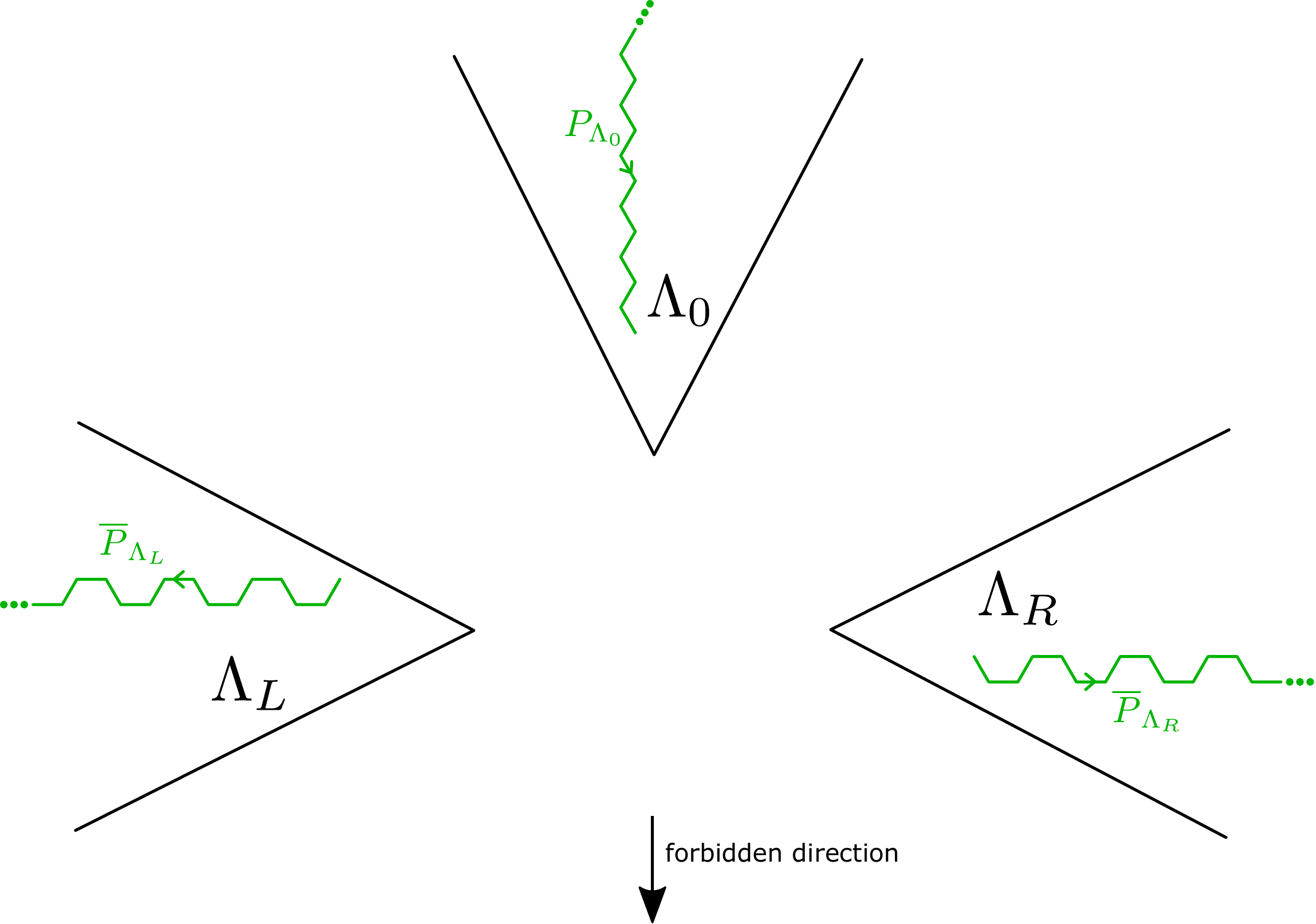}
\caption{Cones $\Lambda_0$, $\Lambda_L$ and $\Lambda_R$ used to define the braiding intertwiners.}
\label{fig:semion braiding setup}
\end{figure}

It follows from Lemma \ref{lem:semion transport lemma} and Proposition \ref{prop:transportability} that there are unitaries $U_a \in (w_a, w^L_a)$ and $V_a \in (w_a, w^R_a)$ which are unique up to phase. By definition \ref{def:braiding},
\begin{equation}
	\ep(a, b) = V_b^* \, \overline w_a \big( V_b \big).
\end{equation}

In order to compute $\overline w_a(V_b)$, let us realise $V_b$ as the weak limit of a sequence of strictly local unitaries.

Let $K$ be the cone whose legs coincide with the central axes of $\Lambda_0$ and $\Lambda_R$, see Figure \ref{fig:braiding limit}. Then the path $\partial \Pi_{K}$ contains $P_{\Lambda_0}$ and $\overline P_{\Lambda_R}$ and the path $Q = \partial \Pi_{K} \setminus ( P_{\Lambda_0} \cup P_{\Lambda_R} )$ is finite. For each $n$ let $K_n = K \cap B_n$ where $B_n \subset \R^2$ is the disk of radius $n$ centered at the origin of $\R^2$. Consider the sequence of paths $P_n = \partial \Pi_{K_n} \setminus Q$ and set $V_b^{(n)} := W_b[P_n]$.

\begin{figure}
\centering
\includegraphics[width = 0.6\textwidth]{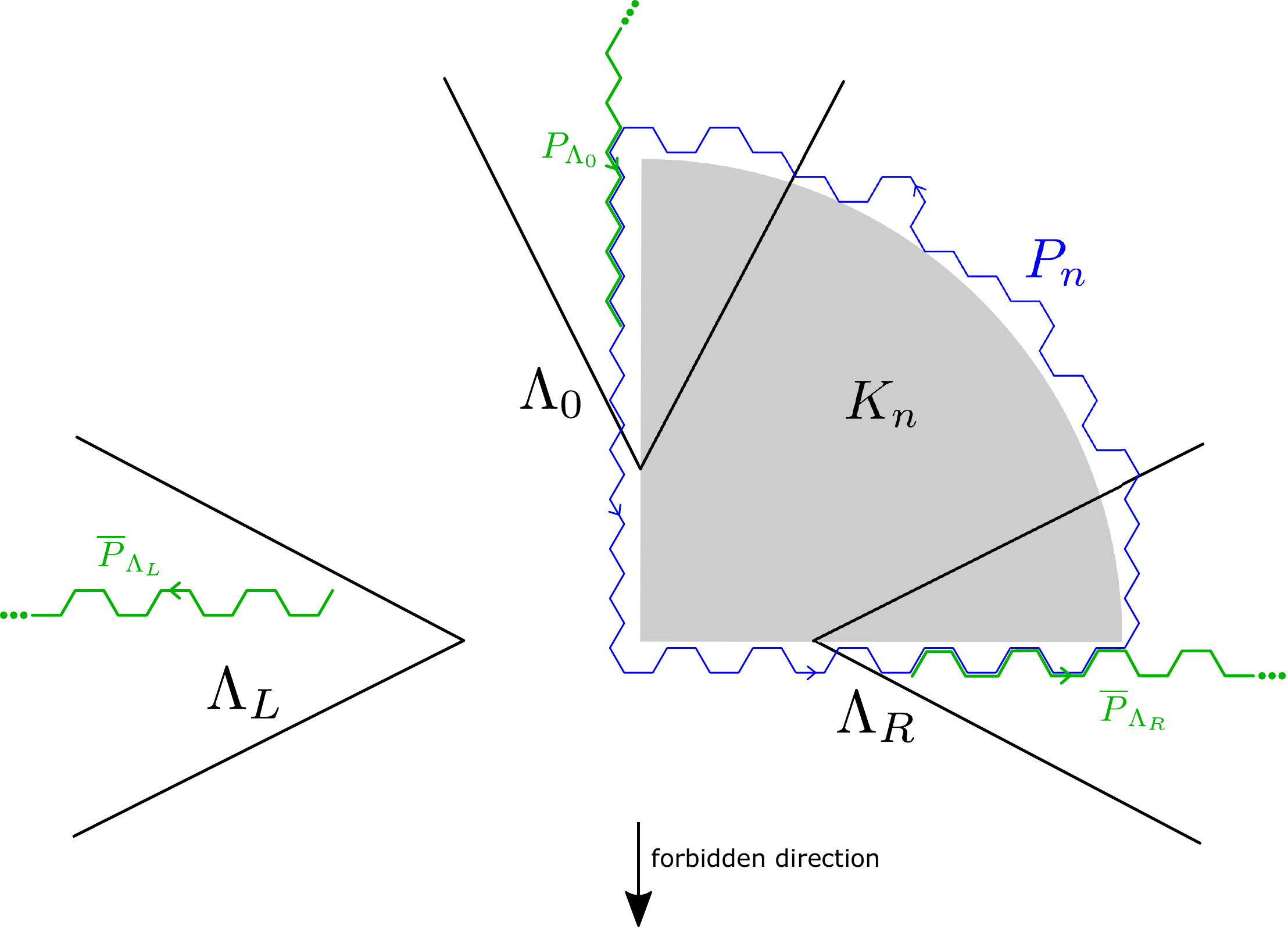}
\caption{Sets $K_n$ and their boundary paths $P_n$ used in the construction of the sequence of unitaries $V_b^{(n)}$ that converge weakly to the intertwiner $V_b$.}
\label{fig:braiding limit}
\end{figure}

\begin{lemma}
	The sequence $V_b^{(n)}$ converges weakly to $V_b$.
\end{lemma}

\begin{proof}
	Consider first the sequence of closed paths $\widetilde P_n = \partial \Pi_{K_n}$ and corresponding string operators $W_b[\widetilde P_n]$. By Proposition \ref{prop:closed strings leave the ground state invariant}, the unitaries $W_b[\widetilde P_n]$ leave the ground state invariant up to a phase, so there are phases $\phi_n$ such that $\widetilde V^{(n)} := \phi_n W_b[\widetilde P_n]$ satisfy $\widetilde V^{(n)} \Omega = \Omega$. Moreover, the $\widetilde V^{(n)}$ differ from $V_b^{(n)}$ by a unitary $W = V_b^{(n)}  (\widetilde V^{(n)})^*$ that does not depend on $n$ (for $n$ large enough, and possible redefining the phases of the $V_b^{(n)}$), and is supported near the path $Q$.

	Since $\partial \Pi_K$ is a closed path, the automorphism $w_b[\partial \Pi_{K}]$ leaves the grounds state invariant by Proposition \ref{prop:closed strings leave the ground state invariant}, so there is a unique unitary $\widetilde V \in \caB(\caH)$ such that $\widetilde V \Omega = \Omega$ and $w_b[\partial \Pi_K] = \Ad(\widetilde V)$ (as automorphisms on $\pi_1(\caA)$).

	Now, for strictly local operators $O, O'$ we have
	\begin{equation}
		\langle O \Omega, \widetilde V \, O' \Omega \rangle = \langle O' \Omega, w_b[\partial \Pi_K]( O' ) \, \widetilde V \Omega \rangle = \lim_{n \uparrow \infty} \langle O \Omega, w_b[\widetilde P_n](O') \, \Omega \rangle = \lim_{n \uparrow \infty} \langle O \Omega, \widetilde V^{(n)} O' \Omega \rangle,
	\end{equation}
	showing that the sequence $\widetilde V^{(n)}$ converges weakly to $\widetilde V$. It follows that the sequence $V_b^{(n)} = W \widetilde V^{(n)}$ converges weakly to $W \widetilde V$. By construction, $\Ad(W \widetilde V) = w_b \circ (w_b^R)^{-1} = \Ad(V_b)$ and both $W \widetilde V \Omega = \Omega$ and $V_b \Omega = \Omega$. Hence $V_b = W \widetilde V$. We conclude that the sequence $V_n^{(n)}$ converges weakly to $V_b$.
\end{proof}

We can now compute the braiding intertwiners.

Obviously $\overline w_1 = \id$ and $V_1 = I$ so  $\ep(1, a) = \ep(a, 1) = \I$ for all $a \in I$. It is also easy to see that
\begin{equation}
	w_{B}(V_a^{(n)}) = V_a^{(n)}
\end{equation}
for any $a \in \{1, S, \bar S, B\}$, so $\ep(B, a) = \I$ for all $a$, while
\begin{equation}
	w_S(V_B^{(n)}) = w_{\bar S}(V_B^{(n)}) = -V_B^{(n)},
\end{equation}
because the path $P_{\Lambda_0}$ contains a single R-leg of the path $P_n$. So $\ep(S, B) = \ep(\bar S, B) = -\I$.

Let us now compute $\ep(S, S)$. Note that the path $P_n$ enters the path $P_{\Lambda_0}$ at an L-vertex of $P_{\Lambda_0}$. Let $(i, j)$ be the edges of $P_{\Lambda_0}$ before and after this L-vertex, see Figure \ref{fig:intersection}. We find
\begin{align*}
	(V_S^{(n)})^* \, w_S( V_S^{(n)} ) &= \left( \sigma_j^X \left( \iu^{\frac{1 - \sigma_i^Z}{2}} \right) \right) \left( \sigma_i^X \sigma_j^X (-1)^{s_I} \right) \left( \left( \iu^{\frac{1 - \sigma_i^Z}{2}} \right)^* \sigma_j^X  \right) \left(  (-1)^{s_I} \sigma_j^X \sigma_i^X  \right) \\
					&= \iu \I,
\end{align*}
which implies $\ep(S, S) = \iu \I$.

\begin{figure}
\centering
\includegraphics[width = 0.2\textwidth]{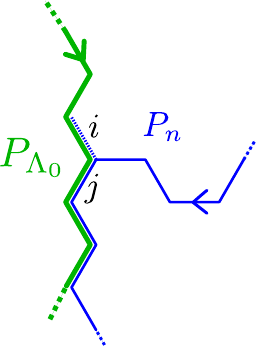}
\caption{Edges $i$ and $j$ playing a role in the computation of $\ep(S, S)$.}
\label{fig:intersection}
\end{figure}

We now use the braid equations (Lemma \ref{lem:braid equations})
\begin{align*}
	\ep(\rho, \sigma \otimes \tau) &= (\I_{\sigma} \otimes \ep(\rho, \tau))(\ep(\rho, \sigma) \otimes \I_{\tau}) \\
	\ep(\rho \otimes \sigma, \tau) &= (\ep(\rho, \tau) \otimes \I_{\sigma}) (\I_{\rho}) \otimes \ep(\sigma, \tau) )
\end{align*}
(where $I_{\rho}$ denotes the identity intertwiner from $\rho$ to itself) to compute
\begin{align*}
	\ep(\bar S, S) &= \ep(S \times B, S) = \ep(S, S) \overline w_S( \ep(B, S) ) = \iu \I \\
	\ep(S, \bar S) &= \ep(S, S \times B) = \overline w_S(\ep(S, B)) \ep(S, S) = -\iu \I \\
	\ep(\bar S, \bar S) &= \ep(S \times B, \bar S) = \ep(S, \bar S) \overline w_S(\ep(B, S)) = - \iu \I.
\end{align*}
Thus we have computed all braiding intertwiners,see Table \ref{tab:braiding intertwiners} for a summary.

\begin{table} 
	\begin{center} 
	\begin{tabular}{ |c|c c c c | } 
		\hline
		$\ep(a, b)$ & $1$ & $S$ & $\bar S$ & $B$ \\
		\hline
		$1$ & $\I$ & $\I$ & $\I$ & $\I$ \\
		$S$ & $\I$ & $\iu \I$ & $-\iu \I$ & $-\I$ \\
		$\bar S$ & $\I$ & $\iu \I$ & $-\iu \I$ & $-\I$ \\
		$B$ & $\I$ & $\I$ & $\I$ & $\I$ \\
		\hline
	\end{tabular}
	\end{center}
	\caption{The braiding intertwiners $\ep(a, b)$ for the double semion state.}
	\label{tab:braiding intertwiners}
\end{table}

The $R$-symbols are defined \eqref{eq:R-symbols defined} by
\begin{equation}
	 \Omega(b, a) \ep(a, b) = R(a, b) \times \Omega(a, b).
\end{equation}

Since $\Omega(a, b) = \Omega(b, a)$ for all $a, b$, we find that the $R$-symbols are as in table \ref{tab:R-symbols}.

\begin{table} 
	\begin{center} 
	\begin{tabular}{ |c|c c c c | } 
		\hline
		$R(a, b)$ & $1$ & $S$ & $\bar S$ & $B$ \\
		\hline
		$1$ & $1$ & $1$ & $1$ & $1$ \\
		$S$ & $1$ & $\iu$ & $-\iu$ & $-1$ \\
		$\bar S$ & $1$ & $\iu$ & $-\iu$ & $-1$ \\
		$B$ & $1$ & $1$ & $1$ & $1$ \\
		\hline
	\end{tabular}
	\end{center}
	\caption{The $R$-symbols $R(a, b)$ for the double semion state.}
	\label{tab:R-symbols}
\end{table}

One can verify that the $F$ and $R$-symbols indeed satisfy the pentagon and hexagon equations.

%%%%%%%%%%%%%%%%%%%%%%%%%%%%%%%%%%%%%%%%%%%%%%%%%%%%%%%%%%%%%%%%%%%%%%%%%%%%%%%%%%%
%%%%%%%%%%%%%%%%%%%%%%%%%%%%%%%%%%%%%%%%%%%%%%%%%%%%%%%%%%%%%%%%%%%%%%%%%%%%%%%%%%%
\subsection{Equivalence to $\mathrm{Rep}_f \caD^{\phi}(\Z_2)$}

%%%%%%%%%%%%%%%%%%%%%%%%%%%%%%%%%%%%%%%%%%%%%%%%%%%%%%%%%%%%%%%%%%%%%%%%%%%%%%%%%%%
\subsubsection{The braided fusion category of anyons}

Let $\caC$ be the category whose objects are the four anyons types $\{1, s, \bar s, b\}$ seen as one-dimensional vector spaces over $\C$. The monoidal structure is determined by the fusion rules:
\begin{equation}
	a \otimes b = a \times b
\end{equation}
which has unit object $1$, and the associators are given by the $F$-symbols. \ie $\al_{a, b, c} : (a \otimes b) \otimes c \rightarrow a \otimes (b \otimes c)$ is given by multiplication with $F(a, b, c)$.

The braiding intertwiners are given by the $R$-symbols, \ie $\ep_{a, b} : a \otimes b \rightarrow b \otimes a$ is given by multiplication with $R(a, b)$.

Finally, each object is its own dual $a^* = a$ with evaluation maps $\ev_a : a \otimes a^* = 1 \rightarrow 1$ given by multiplication with $-1$ if $a \in \{S, \bar S\}$ and multiplication by $1$ if $a \in \{1, b \}$, and coevaluation maps $i_a : 1 \rightarrow a \otimes a^* = 1$ given by multiplication by $1$.

Since the $F$ and $R$-symbols satisfy the pentagon and hexagon equations, this data describes a braided fusion category. In the rest of this section, we show that this is precisely the category of representations of the quasi quantum double $\caD^{\phi}(\Z_2)$ where $\phi$ is a representative of the non-trivial class of $H^3(\Z_2, U(1))$.

%%%%%%%%%%%%%%%%%%%%%%%%%%%%%%%%%%%%%%%%%%%%%%%%%%%%%%%%%%%%%%%%%%%%%%%%%%%%%%%%%%%
\subsubsection{Description of $\mathrm{Rep}_f \caD^{\phi}(\Z_2)$}

We describe a quasi Hopf algebra $\caD^{\phi}(\Z_2)$ first introduced in \cite{dijkgraaf1991quasi}. We follow the presentation in \cite{propitius1995topological}.

Let $\phi : (\Z_2)^3 \rightarrow U(1)$ be the normalized representative of the non-trivial class in $H^3(\Z_2, U(1))$:
\begin{equation}
	\phi(-, -, -) = -1, \quad \text{all other components equal to 1}.
\end{equation}

Let
\begin{equation}
	c_f(g, h) := (\iota_f \phi)(g, h) = \frac{\phi(f, g, h)\phi(g, h, f)}{\phi(g, f, h)}.
\end{equation}
For each $f$, the map $c_f : (\Z_2)^2 \rightarrow U(1)$ is a 2-cocycle, it satisfies
\begin{equation} \label{eq:cocycle condition for c_x}
	c_x(f, g) c_x(fg, h) = c_x(f, gh) c_x(fg, h).
\end{equation}

The quasi quantum double $\caD^{\phi}(\Z_2)$ is an algebra spanned by $\{ P_x f \}_{x, f \in \Z_2}$ with multiplication
\begin{equation}
	(P_x f) (P_y g) = \delta_{x, y} (P_x fg)  \, c_x(f, g).
\end{equation}
The unit for this multiplication is $\sum_{x \in \Z_2} (P_x 1)$. 

The quasi quantum double is moreover equipped with a coproduct $\Delta : \caD^{\phi}(\Z_2) \rightarrow \caD^{\phi}(\Z_2) \otimes \caD^{\phi}(\Z_2)$ given by
\begin{equation}
	\Delta(P_x f) = \sum_{y z = x}  c_f(y, z)  (P_y f) \otimes (P_z f).
\end{equation}

Associativity and coassociativity follow readily from Eq. \eqref{eq:cocycle condition for c_x}. That $\Delta$ is an algebra morphism follows from the identity
\begin{equation}
	\frac{c_x(f, g) c_y(f, g)}{c_{xy}(f, g)} \times \frac{c_f(x, y) c_g(x, y)}{c_{fg}(x, y)} = 1.
\end{equation}

There is a counit $\ep : \caD^{\phi}(\Z_2) \rightarrow \C$ and an antipode $S : \caD^{\phi}(\Z_2) \rightarrow \caD^{\phi}(\Z_2)$:
\begin{equation}
	\ep( P_x f ) = \delta_{x, 1}  (P_1 1), \quad S( P_x f) = (P_{x^{-1}} f^{-1}) \, c_{x^{-1}}(f, f^{-1})^{-1} c_f(x, x^{-1})^{-1}.
\end{equation}
These give $\caD^{\phi}(\Z_2)$ the structure of a quasi Hopf algebra.

This quasi Hopf algebra is moreover quasitriangular with universal R-matrix
\begin{equation}
	R = \sum_{x, y} (P_x 1 ) \otimes (P_y x).
\end{equation}

%%%%%%%%%%%%%%%%%%%%%%%%%%%%%%%%%%%%%%%%%%%%%%%%%%%%%%%%%%%%%%%%%%%%%%%%%%%%%%%%%%%
\subsubsection{Category of representations}

There are four irreducible representations of $\caD^{\phi}(\Z_2)$, labeled by pairs $(x, \chi) \in \Z_2 \times \Z_2^*$. ($\Z_2^*$ consists of the characters of $\Z_2$, namely $1$ and $\sgn$.) They are given by
\begin{equation}
	\Pi_{(x, \chi)}(P_y f) = \delta_{x, y} \, \vep_x(f) \chi(f)
\end{equation}
with
\begin{equation}
	\vep_x(f) := \exp \left( \frac{\pi \iu }{2} [x] . [f] \right)
\end{equation}
where $[x], [f]$ are the additive representation of $x$ and $f$. \ie $\vep_{-}(-) = \iu$ and all other components are equal to one. $\vep_x$ is a cocycle and
\begin{equation}
	c_x(f, g) = (d \vep_x)(f, g) = \frac{\vep_{x}(fg)}{\vep_x(f) \vep_x(g)}.
\end{equation}

Since we have a coproduct, we can define the tensor product of representations:
\begin{equation}
	(\Pi_1 \otimes \Pi_2)(P_x f) := ((\Pi_1 \otimes \Pi_2) \circ \Delta)(P_x f) = \sum_{yz = x} c_f(y, z) \, \Pi_1(P_y f) \otimes \Pi_2(P_z f).
\end{equation}
One easily verifies that
\begin{equation} \label{eq:fusion of representations}
	\Pi_{(x, \chi)} \otimes \Pi_{(y, \sigma)} = \Pi_{(xy, \chi \sigma)}.
\end{equation}
The representation $\Pi_{(1, 1)}$ is an identity for this tensor product (with trivial left and right unitors).

With this tensor product, the representations of $\caD^{\phi}(\Z_2)$ form a tensor category with simple objects $\Pi_{(x, \chi)}$ and associators
\begin{equation}	
	\al_{(x, \chi), (y, \sigma), (z, \tau)} : ( \Pi_{(x, \chi)} \otimes \Pi_{(y, \sigma)} ) \otimes \Pi_{(z, \tau)} \rightarrow  \Pi_{(x, \chi)} \otimes (\Pi_{(y, \sigma)} \otimes \Pi_{(z, \tau)})
\end{equation}
given by multiplication with $\phi^{-1}(x, y, z)$. Since $\phi$ is a cocycle this associator satisfies the pentagon identity.

Let us now introduce the braiding on irreducible objects. These are natural tranformations $\ep_{(x, \chi), (y, \sigma)} : \Pi_{(x, \chi)} \otimes \Pi_{(y, \sigma)} \rightarrow \Pi_{(y, \sigma)} \otimes \Pi_{(x, \chi)}$. They are given by multiplication with
\begin{equation}
	(\Pi_{(x, \chi)} \otimes \Pi_{(y, \sigma)}) (R) = \varepsilon_y(x) \sigma(x).
\end{equation}

These braidings are summarised in table \ref{tab:braiding of representations}.

\begin{table} 
	\begin{center} 
	\begin{tabular}{ |c|c c c c | } 
		\hline
		$\ep_{(x, \chi), (y, \sigma)}$ & $(1, 1)$ & $(-1, 1)$ & $(-1, \sgn)$ & $(1, \sgn)$ \\
		\hline
		$(1, 1)$ & $1$ & $1$ & $1$ & $1$ \\
		$(-1, 1)$ & $1$ & $\iu$ & $-\iu$ & $-1$ \\
		$(-1, \sgn)$ & $1$ & $\iu$ & $-\iu$ & $-1$ \\
		$(1, \sgn)$ & $1$ & $1$ & $1$ & $1$ \\
		\hline
	\end{tabular}
	\end{center}
	\caption{The braiding isomorphisms of $\mathrm{Rep}_f \caD^{\phi}(\Z_2)$.}
	\label{tab:braiding of representations}
\end{table}

%%%%%%%%%%%%%%%%%%%%%%%%%%%%%%%%%%%%%%%%%%%%%%%%%%%%%%%%%%%%%%%%%%%%%%%%%%%%%%%%%%%
\subsubsection{Braided monoidal equivalence}

Comparing tables \ref{tab:braiding intertwiners} and \ref{tab:braiding of representations} we are led to identify
\begin{equation}
	(1, 1) \leftrightarrow 1, \quad (-1, 1) \leftrightarrow S, \quad (-1, \sgn) \leftrightarrow \bar S, \quad (1, \sgn) \leftrightarrow B.
\end{equation}
Under this identification, the $F$-symbols also match with the associators $\al$ of the representation category. Thus the braided fusion category $\caC$ and the representation category $\mathrm{Rep}_f \caD^{\phi}(\Z_2)$ are isomorphic as braided monoidal categories.

%%%%%%%%%%%%%%%%%%%%%%%%%%%%%%%%%%%%%%%%%%%%%%%%%%%%%%%%%%%%%%%%%%%%%%%%%%%%%%%%%%%
%%%%%%%%%%%%%%%%%%%%%%%%%%%%%%%%%%%%%%%%%%%%%%%%%%%%%%%%%%%%%%%%%%%%%%%%%%%%%%%%%%%
%%%%%%%%%%%%%%%%%%%%%%%%%%%%%%%%%%%%%%%%%%%%%%%%%%%%%%%%%%%%%%%%%%%%%%%%%%%%%%%%%%%
%%%%%%%%%%%%%%%%%%%%%%%%%%%%%%%%%%%%%%%%%%%%%%%%%%%%%%%%%%%%%%%%%%%%%%%%%%%%%%%%%%%
%%%%%%%%%%%%%%%%%%%%%%%%%%%%%%%%%%%%%%%%%%%%%%%%%%%%%%%%%%%%%%%%%%%%%%%%%%%%%%%%%%%
%%%%%%%%%%%%%%%%%%%%%%%%%%%%%%%%%%%%%%%%%%%%%%%%%%%%%%%%%%%%%%%%%%%%%%%%%%%%%%%%%%%
\appendix

%%%%%%%%%%%%%%%%%%%%%%%%%%%%%%%%%%%%%%%%%%%%%%%%%%%%%%%%%%%%%%%%%%%%%%%%%%%%%%%%%%%
%%%%%%%%%%%%%%%%%%%%%%%%%%%%%%%%%%%%%%%%%%%%%%%%%%%%%%%%%%%%%%%%%%%%%%%%%%%%%%%%%%%
%%%%%%%%%%%%%%%%%%%%%%%%%%%%%%%%%%%%%%%%%%%%%%%%%%%%%%%%%%%%%%%%%%%%%%%%%%%%%%%%%%%
\section{Purity of the double semion state} \label{app:purity}

In this appendix we prove Theorem \ref{thm:purity}, stating that the double semion state $\omega$ constructed in Section \ref{sec:state construction} is pure.

%%%%%%%%%%%%%%%%%%%%%%%%%%%%%%%%%%%%%%%%%%%%%%%%%%%%%%%%%%%%%%%%%%%%%%%%%%%%%%%%%%%
%%%%%%%%%%%%%%%%%%%%%%%%%%%%%%%%%%%%%%%%%%%%%%%%%%%%%%%%%%%%%%%%%%%%%%%%%%%%%%%%%%%
\subsection{Restrictions of $\omega$ to finite regions}

Denote by $\overline \Pi_n$ the set of edges belonging to some plaquette of $\Pi_n$ and set $\caM_n = \caA_{\overline \Pi_n}$. We investigate the restrictions $\omega|_n := \omega|_{\caM_n}$.

\begin{lemma} \label{lem:restrictions look like omega in the bulk}
	For any $m > n$ we have that $\omega|_n = \omega_m|_n$.
\end{lemma}

\begin{proof}
	It is sufficient to note that $\omega_m$ has the same expectation value for any operator in $\caM_n$ as $\omega$ does. Indeed, for any $A \in \caM_n$ we have $\omega(A) = \lim \omega_m(A)$, and the latter sequence becomes constant as soon as $\Pi_m$ contains all plaquettes containing edges in the support of $A$. \ie $\omega_m(A) = \omega(A)$ for all $m \geq n + 1$.
\end{proof}

Thus we can restrict our attention to states $\omega_{m}|_n$. Recall that $\omega_m$ is given by the expectation value in the the vector state
\begin{equation}
	\Omega_m = \sqrt{\frac{1}{2^{\abs{\Pi_m}}}}  \sum_{\Pi \subset \Pi_m} (-)^{\sharp\Pi} A_{\Pi} \Omega_0 = \sqrt{ \frac{1}{2^{\abs{\Pi_m}}} } \sum_{\Pi \subset \Pi_m} (-)^{\sharp \partial \Pi} | \partial \Pi \rangle
\end{equation}
where $\sharp \partial \Pi$ is the number of closed loops in the loop soup $\partial \Pi$, and we chose to write $A_{\Pi}$ instead of $\pi(A_{\Pi})$ because the representation $\pi$ is faithful, and $| \partial \Pi \rangle$ is the product state with all degrees of freedom spin up, except those on the edges along the dual path $\partial \Pi$, which are spin down.

Note that every closed $\al$ supported on $\overline \Pi_m$ is of the form $\partial \Pi$ for a unique $\Pi \subset \Pi_m$, so we have written $\Omega_m$ as a uniform superposition over all closed loop soups supported on $\overline \Pi_m$. Moreover, for closed loops $\al$ and $\beta$ we have $| \al \rangle = | \beta \rangle$ if and only if $\al = \beta$, and these states are orthogonal otherwise.

We will show that $\omega_m|_n$ is a mixed state which is an equal-weight convex combination of pure states $\eta_n(b)$ where $b$ is a \emph{boundary condition}, namely an assignment of up-or down to each out edge of the region $\Pi_{n}$ such that an even number of edges are up, see Figure \ref{fig:boundary conditions}

The state $\eta_n(b)$ is then given by a uniform superposition of all loop soups that satisfy the boundary condition $b$, weighed by $\pm 1$ depending on whether a fixed `closure' of the boundary condition has an even or an odd number of closed loops.
\begin{lemma}
	There are $2^{\abs{\Pi_n}}$ such loop soups for each boundary condition $b$.
\end{lemma}

\begin{proof}
	For the boundary condition with all spins up this is obvious, because then the loop soups are precisely closed loop soups in $\overline \Pi_n$.

	To obtain loop soups for an arbitrary boundary condition $b$, act on any closed loop soup with $A_p$ on the plaquettes between pairs of boundary edges where $b$ forces a loop to end (choose one of two possible pairings). This yields a loop soup that satisfies the boundary condition, and two different closed loop soups give two different loop soups satisfying the boundary condition. Conversely, evey loop soup satisfying the boundary condition arises in this way, because acting on loop soups satisfying $b$ with $A_v$'s on the vertices between pairs of boundary edges where $b$ forces a loop to end yields a closed loop soup.
\end{proof}

Write $\caP_n^{(b)}$ for the loop soups in $\overline \Pi_n$ that satisfy the boundary condition $b$. For a given boundary condition $b$ any $\al \in \caP_n^{(b)}$ can be `closed up' by connecting neighbouring marked edges in two ways, see figure \ref{fig:boundary conditions}. Pick one such `pairing' of marked boundary edges and let $\sharp \al$ be the number of loops of $\al$ closed up with the chosen pairing. Then we have normalized vectors
\begin{equation}
	|\eta_n^{(b)} \rangle = \sqrt{\frac{1}{2^{\abs{\Pi_n}}}} \sum_{\al \in \caP_n^{(b)}} (-1)^{\sharp \al} | \al \rangle.
\end{equation}
We have $\langle \eta_n^{(b)}, \eta_n^{(b')} \rangle = \delta_{b, b'}$, \ie these vectors form an orthonormal set. Denote by $\eta_n^{(b)}$ the pure state on $\caM_n$ corresponding to the vector $| \eta_n^{(b)} \rangle$.

\begin{figure}
\centering
\includegraphics[width = 0.4\textwidth]{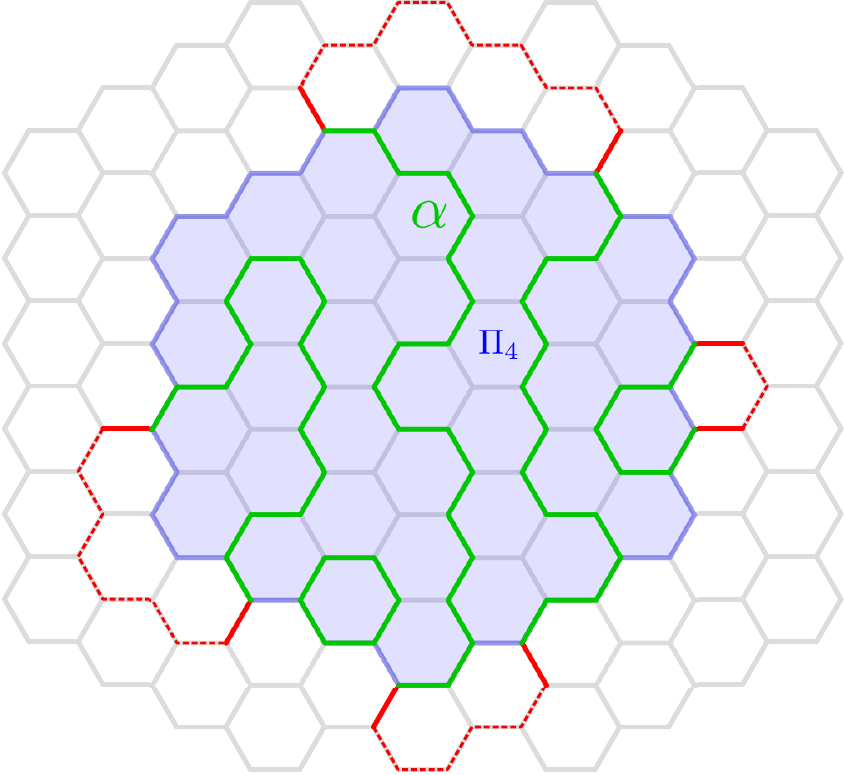}
\caption{A loop soup $\al \in \caP_4^{(b)}$ with boundary condition $b$ corresponding to the red edges. The dotted red paths indicate one of two ways of pairing neighbouring red edges, resulting in a closed loop soup.}
\label{fig:boundary conditions}
\end{figure}

\begin{proposition} \label{prop:Schmidt decomposition of restriction}
	For $m > n \geq 1$,
	\begin{equation}
		\omega_m|_n = \frac{1}{2^{6n + 1}} \sum_{b} \eta_n^{(b)}.
	\end{equation}
\end{proposition}
Since the $| \eta_n^{(b)} \rangle$ form an orthonormal set, this is a Schmidt decomposition of $\omega_m|_n$.

Here, $2^{6n - 1}$ is the number of boundary conditions $b$. Indeed, there are $6n$ outer edges where the boundary condition either forces or does not force a string to pass, and the number of edges where a string is forced to end must be even. There are as many even boundary conditions as there are odd boundary conditions. Indeed, flipping a fixed edge gives a bijection.

\begin{proof}
	By Lemma \ref{lem:restrictions look like omega in the bulk} it is sufficient to consider $m = n+1$. The state $\omega_{n+1}$ on $\overline \Pi_{n+1}$ is a uniform superpostition of closed loop soups in $\overline \Pi_{n+1}$. Any such loop soup $\al$ defines a boundary condition $b(\al)$ by the outer edges of $\Pi_n$ that are occupied by strings of $\al$. We can therefore organize the $\al$ according to which boundary condition they induce:
	\begin{equation}
		|\Omega_{n+1} \rangle = \sqrt{\frac{1}{2^{\abs{\Pi_{n+1}}}}}  \sum_{b} \sum_{\al : b(\al) = b} \, (-)^{\sharp \al}  | \al \rangle.
	\end{equation}
	The states $| \al \rangle$ are orthonormal product states. If $O$ is supported on $\overline \Pi_n$ then the matrix elements $\langle \beta, O \al \rangle$ only depend on the configuration of $\al$ and $\beta$ on $\overline \Pi_n$. This information still allows us to deduce the boundary conditions $b(\al)$ and $b(\beta)$. Moreover, the matrix element vanishes if $b(\al) \neq b(\beta)$, hence
	\begin{align*}
		\omega_{n+1}(O) = \langle \Omega_{n+1}, O \Omega_{n+1} \rangle &= \frac{1}{2^{\abs{\Pi_{n+1}}}} \sum_{\al, \beta} (-)^{\sharp \al + \sharp \beta} \langle \beta, O \al \rangle \\
						     &= \frac{1}{2^{\abs{\Pi_{n+1}}}} \sum_{b} \sum_{\substack{\al : b(\al) = b \\ \beta : b(\beta) = b}} (-)^{\sharp \al + \sharp \beta} \langle \beta, O \al \rangle \\
						     &= \frac{2}{2^{\abs{\Pi_{n+1}}}}  \sum_{b} \sum_{\al', \beta' \in \caP_n^{(b)}} (-)^{\sharp \al' + \sharp \beta'} \langle \beta', O \al' \rangle \\
						     &= \frac{1}{2^{6n - 1}} \sum_b \, \frac{1}{2^{\abs{\Pi_n}}} \sum_{\al', \beta' \in \caP_n^{(b)}} (-)^{\sharp \al' + \sharp \beta'}  \langle \beta', O \al' \rangle \\
						     &= \frac{1}{2^{6n - 1}} \sum_b \eta_n^{(b)}(O).
	\end{align*}
	The factor of $2$ appearing in the third line is the number of choices of completing a loop soup $\al'$ in $\overline \Pi_{n}$ with boundary condition $b$ to a closed loop soup $\al$ in $\overline\Pi_{n+1}$. The phase $(-)^{\sharp \al' + \sharp \beta'}$ does not depend on which (common) completion is chosen. Indeed, changing the completion changes both $\sharp \al'$ and $\sharp \beta'$ by an odd amount if the number of marked edges is a multiple of 4, and both by an even amount otherwise (Lemma \ref{lem:dependence on pairing}). To get the fourth line we used that $\abs{\Pi_{n+1}} - \abs{\Pi_n} = 6n - 1$.
\end{proof}

\begin{lemma} \label{lem:dependence on pairing}
	Given $\al \in \caP_n^{(b)}$, denote by $\sharp_1 \al$ and $\sharp_2 \al$ the number of loops in the two possible completions. then $\sharp_1 \al - \sharp_2 \al$ is odd if the number of marked points for $b$ is a multiple of 4, and even otherwise.
\end{lemma}

\begin{proof}
	Assume first that $\al$ has no closed loops. Let the number of marked points be $2n$. The following construction is illustrated in Figure \ref{fig:graph}. Abstract the region $\Pi_n$ to a disk with the marked points sitting on the boundary. Then the two completions correspond to two sets of $n$ intervals that `interlace' along the boundary of the disk. Choose one of them. The loop soup $\al$ connects these $n$ intervals into groups. The number of groups $g$ is the number of closed loops in this completion, say $\sharp_1 \al = g$. Put a vertex on each interval for this completion, and add a vertex in each group. Connect this vertex by edges to the vertices of the intervals in the group. Finally, connect the vertices on the intervals by edges along the boundary of the disk. This gives a connected graph with $V = n + g$ vertices and $E = 2n$ edges. By the Euler formula, this graph has $F = 1-V+E = 1 + n - g$ internal faces. The number of internal faces corresponds precisely to $\sharp_2 \al$, and we find
	\begin{equation}
		\sharp_1 \al - \sharp_2 \al = g - (1 + n - g) = 2g -n - 1,
	\end{equation}
	which is odd if $n$ is even and vice versa.

	Any closed loops of $\al$ remain connected components of both completions, so internal loops do not contribute to $\sharp_1 \al = \sharp_2 \al$.

	\begin{figure}
	\centering
	\includegraphics[width = 0.4\textwidth]{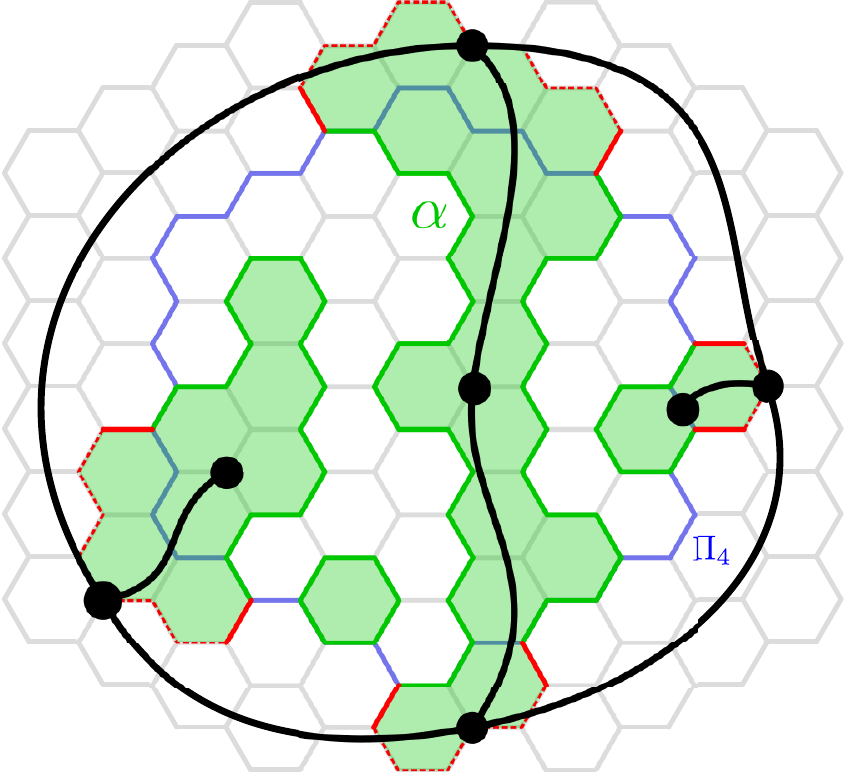}
	\caption{The red dotted paths completing $\al$ to a closed loop soup are marked with vertices (black), and so are the closed regions (green) resulting from this completion. The white regions correspond one-to-one to faces of the black graph. Each such white region corresponds to a loop of the alternative completion of $\al$ to a closed loop soup.}
	\label{fig:graph}
	\end{figure}
\end{proof}

We further show
\begin{lemma} \label{lem:etas look like omega in the bulk}
	For any boundary condition $b$ and any $O$ supported on $\overline \Pi_{n-1}$ we have $\eta_{n}^{(b)}(O) = \omega(O)$.
\end{lemma}

\begin{proof}
	From Lemma \ref{lem:restrictions look like omega in the bulk}, it is sufficient to show that $\eta_n^{(b)}(O) = \omega_n(O)$. (I write $\omega_n$ for the restriction of this state to $\overline \Pi_n$). Note that $\omega_n = \eta_n^{\emptyset}$, where $\emptyset$ stand for the trivial boundary condition.

	For any other boundary condition $b$, let $A_b$ be the product of $A_p$ operators over plaquettes between pairs of marked edges of $b$. Clearly, $A_b$ is supported outside $\overline \Pi_{n-1}$, so $A_b^* O A_b = O$ and since $A_b$ bijectively maps loop soups satisfying $b$ to closed loop soups, we find
	\begin{equation}
		\eta_n^{(b)}(O) = \eta_n^{(b)}(A_b^* O A_b) = \eta_n^{\emptyset}(O) = \omega_n(O) = \omega(O).
	\end{equation}
\end{proof}

%%%%%%%%%%%%%%%%%%%%%%%%%%%%%%%%%%%%%%%%%%%%%%%%%%%%%%%%%%%%%%%%%%%%%%%%%%%%%%%%%%%
%%%%%%%%%%%%%%%%%%%%%%%%%%%%%%%%%%%%%%%%%%%%%%%%%%%%%%%%%%%%%%%%%%%%%%%%%%%%%%%%%%%
\subsection{Purity of the limit state}

We will now show that $\omega$ is a pure state by making use of the following Lemma, which is a special case of Lemma 2.1. of \cite{glimm1960certain}.
\begin{lemma} [Lemma 2.1 of \cite{glimm1960certain}] \label{lem:Glimm}
	A state $\omega$ on a UHF algebra realized as the inductive limit of a sequence of finite matrix algebras $\{\caM_m\}$ is pure if the following holds:

	For each $n$ there exists $m > n$ such that if $\rho$ is a linear functional on $\caM_m$ that satifies
	\begin{equation} \label{eq:omega majorizes rho}
		\omega|_{\caM_m} \geq \rho \geq 0,
	\end{equation}
	then
	\begin{equation}
		\rho|_{\caM_n} = \lambda \omega|_{\caM_n}
	\end{equation}
	for some $\lambda \in \R$.
\end{lemma}

In applying this Theorem to our setting, we take $\caM_n$ to be the algebra supported on $\overline \Pi_n$.

Fix $n$ and take $m \geq n+1$. Let $\rho$ be a linear functional on $\caM_m$ such that Eq. \eqref{eq:omega majorizes rho} is satisfied. From Proposition \ref{prop:Schmidt decomposition of restriction}  and Lemma \ref{lem:restrictions look like omega in the bulk} we have
\begin{equation}
	\omega|_{\caM_m} = \omega|_m = \omega_{m+1}|_m = \frac{1}{2^{6m - 1}} \sum_b \eta_m^{(b)},
\end{equation}
which is a Schmidt decomposition for $\omega|_m$. The assumption $\omega|_m \geq \rho \geq 0$ implies that $\rho$ is a mixture of pure states in the span of the $\eta_m^{(b)}$'s (Lemma \ref{lem:dominated schmidt spans subspace of dominating schmidt}). It then follows from Lemma \ref{lem:etas look like omega in the bulk} and $m > n$ that
\begin{equation}
	\rho|_n =  \lambda \omega|_n
\end{equation}
for some $0 \leq \lambda \leq 1$.

We conclude by Lemma \ref{lem:Glimm} that $\omega$ is pure.

We have used the following Lemma:
\begin{lemma} \label{lem:dominated schmidt spans subspace of dominating schmidt}
	Let $\omega$ and $\rho$ be linear functionals on a finite matrix algebra such that $\omega \geq \rho \geq 0$. Suppose
	\begin{equation}
		\omega = \sum_{\al} p_{\al} | \psi_{\al} \rangle \langle \psi_{\al} |
	\end{equation}
	is a Schmidt decomposition of $\omega$. Then any Schmidt decomposition
	\begin{equation}
		\rho = \sum_{\beta} q_{\beta} | \phi_{\beta} \rangle \langle \phi_{\beta} |
	\end{equation}
	satisfies
	\begin{equation}
		\Span \{ | \phi_{\beta} \rangle \}_{\beta} \subset \Span \{ | \psi_{\al} \rangle \}_{\al}.
	\end{equation}
	\ie the Schmidt states of $\rho$ span a subspace of the space spanned by the Schmidt states of $\omega$.
\end{lemma}

\begin{proof}
	Suppose the conlusion is false, so one of the $\phi_{\beta}$, say $\phi_1$, lies outside of $\caV = \Span \{\psi_{\al}\}$. Then there is a vector $\chi$ orthogonal to $\caV$ and such that $c = \langle \phi_1, \chi \rangle \neq 0$.

	Consider now the positive operator $P = | \chi\rangle \langle \chi |$. We have $\omega(P) = 0$ and
	\begin{equation}
		\rho(P) = \sum_{\beta} q_{\beta} |\langle \phi_{\beta}, \chi \rangle|^2 > 0
	\end{equation}
	where all terms are non-negative and at least the term $\beta = 0$ is strictly positive.

	It follows that $(\omega - \rho)(P) < 0$, violating the assumption.
\end{proof}

%%%%%%%%%%%%%%%%%%%%%%%%%%%%%%%%%%%%%%%%%%%%%%%%%%%%%%%%%%%%%%%%%%%%%%%%%%%%%%%%%%%
%%%%%%%%%%%%%%%%%%%%%%%%%%%%%%%%%%%%%%%%%%%%%%%%%%%%%%%%%%%%%%%%%%%%%%%%%%%%%%%%%%%
%%%%%%%%%%%%%%%%%%%%%%%%%%%%%%%%%%%%%%%%%%%%%%%%%%%%%%%%%%%%%%%%%%%%%%%%%%%%%%%%%%%
\section{Properties of string operators} \label{app:strings}

For each vertex $v$ let
\begin{equation}
	A_v := \frac{1}{2} \left( 1 + \prod_{i \sim v} \sigma_i^Z  \right)
\end{equation}
where the product runs over the three edges connected to $v$.
For each hexagon $p$, regarded as a closed loop with counterclockwise orientation, let
\begin{equation}
	B_p := \frac{1}{2} \left( 1 + W_S[p]  \right) \left( \prod_{v \in p} A_v  \right).
\end{equation}
The operators $B_p$ and $A_v$ are orthogonal projections, and they all commute with each other.

Let $\Pi_n$ be a finite set of hexagons as in Figure \ref{fig:increasing sequence}. Let $\overline \Pi_n$ be the set of vertices belonging to some hexagon in $\Pi_n$. We set
\begin{equation}
	H_{\Pi_n} := \sum_{v \in \overline \Pi_n} (1 - A_v) + \sum_{p \in \Pi_n} B_p.
\end{equation}

Let us also introduce terms imposing boundary conditions:
\begin{equation}
	H_{\partial \Pi_n} = \sum_{i \in \partial \overline \Pi_n} \frac{1}{2}\left( 1 - \sigma_i^Z \right).
\end{equation}
This is also a sum of orthogonal projections, and they all commute with each other and with the $B_p$ and $A_v$ appearing in $H_{\Pi_n}$.

We now consider the commuting projection Hamiltonians
\begin{equation}
	 H_n := H_{\Pi_n} + H_{\partial \Pi_n}.
\end{equation}

Let $\widetilde \Pi_n$ be the collection of edges that have an endpoint in $\overline \Pi_n$. Then $\caH_n \in \caA_{\widetilde \Pi_n}$. Moreover, the state $\omega_n$ restricts to $\caA_{\widetilde \Pi_n}$ as a pure state. Let us continue to denote this restriction by $\omega_n$. We have
\begin{lemma} \label{lem:Omega_n is unique ground state}
	 The state $\omega_n$ on $\caA_{\widetilde \Pi_n}$ is the unique ground state of $H_n$.
\end{lemma}

\begin{proof}
 	The state $\omega_n$ is defined by the expectation in the vector state
	\begin{equation}
		\Omega_n = \sqrt{\frac{1}{2^{\abs{\Pi_n}}}}  \, \sum_{\Pi \subset \Pi_n} \, (-1)^{\sharp(\Pi)} \, A_{\Pi} \Omega_0
	\end{equation}
	where $\Omega_0$ has all $\sigma_i^Z = 1$.

	The state $\Omega_n$ is a superposition of closed string configurations in $\Pi_n$. Each such closed string configuration satisfies
	\begin{equation}
		(1 - A_v) A_{\Pi} \Omega_0 = 0, \quad \frac{1}{2}(1 - \sigma_i^Z) A_{\Pi} \Omega_0 = 0
	\end{equation}
	for all $v \in \overline \Pi_n$, all $i \in \partial \overline \Pi_n$ and all $\Pi \subset \Pi_n$.

	To see that $\Omega_n$ is a ground state of $H_n$ it remains to show that it is in the kernel of all $B_p$ for $p \in \Pi_n$.  One can check that
	\begin{equation}
		W_S[p] A_{\Pi} \Omega_0 = \phi(p, \Pi) A_{p \triangle \Pi} \Omega_0
	\end{equation}
	where $\phi(p, \Pi) = -1$ if $p \triangle \Pi$ has the same parity of connected components as $\Pi$, and $\phi(p, \Pi) = 1$ otherwise. \ie
	\begin{equation}
		\phi(p, \Pi) = (-1)^{\sharp \Pi + \sharp(p \triangle \Pi) + 1}
	\end{equation}
	It follows that $W_S[p] \Omega_n = - \Omega_n$ for any $p \in \Pi_n$, hence $B_p \Omega_n = 0$.

	To see that $\Omega_n$ is the unique ground state, observe that any ground state must be in the kernel of all the $1 - A_v$ for $v \in \overline \Pi_n$ and all the  $\frac{1}{2}(1 - \sigma_i^Z)$ for $i \in \partial \overline \Pi_n$. The space of states that are simultaneously in the kernels of all these projections is spanned by the closed string states
	\begin{equation}
		A_{\Pi} \Omega_0, \quad \Pi \subset \Pi_n.
	\end{equation}
	We must find in this space a state that is in the kernel of all the $B_p$, equivalently a -1 eigenstate of all the $W_S[p]$ for $p \in \Pi_n$. Consider a general state
	\begin{equation}
		\Psi = \sum_{\Pi \subset \Pi_n} \psi(\Pi) A_{\Pi} \Omega_0
	\end{equation}
	where $\psi(\Pi) \in \C$ are arbitrary. Then
	\begin{equation}
		W_S[p] \Psi = \sum_{\Pi \subset \Pi_n}  \psi(\Pi) \phi(p, \Pi) A_{p \triangle \Pi} \Omega_0,
	\end{equation}
	so $W_S[p] \Psi = -\Psi$ only if
	\begin{equation}
		\psi(\Pi) (-1)^{\sharp \Pi} - \psi(p \triangle \Pi) (-1)^{\sharp (p \triangle \Pi)}.
	\end{equation}
	If any of the $\psi(\Pi)$ is non-zero (which must be the case, otherwise $\Psi = 0$) then this enforces
	\begin{equation}
		\psi(\Pi') = (-1)^{\sharp \Pi + \sharp \Pi'} \psi(\Pi')
	\end{equation}
	for all $\Pi' \subset \Pi_n$. Indeed, and $\Pi$ can be related to any $\Pi'$ by a sequence of symmetric differences with elementary plaquettes $p$. This shows that $\Psi \simeq \Omega_n$, so $\Omega_n$ is indeed the unique ground state of $H_n$ on $\caA_{\widetilde \Pi_n}$.
\end{proof}

\begin{lemma} \label{lem:strings commute with Hamiltonian}
	If $P$ is a closed path entirely contained in $\Pi_n$, then $W_a[P]$ commutes with $H_n$.
\end{lemma}

\begin{proof}
	This is shown for the string operators of any Levin-Wen model using a graphical representation of the string operators in \cite{levin2005string}. In our case of the double semion model, we can also show it by brute force. That $W_S[P]$ commutes with the star operators $A_v$ and with the boundary terms in $H_{\partial \Pi_n}$ is obvious. Let us show that $W_S[P]$ commutes with $B_p$ for $p \in \Pi_n$.

	To this end, note simply that if $Q$ is the path, possibly consisting of multiple components, made up of edges of $P$ that are also edges or R-legs of $p$, oriented with the same orientation as $P$, then
	\begin{equation}
		W_S[P] W_S[p] W_S[P]^* = W_S[p] V[Q]
	\end{equation}
	where the string operators $V[Q]$ were defined previously in Eq. \eqref{eq:V string defined}. Since $V[Q](\prod_{v \in p} A_v) = \prod_{v \in p} A_v$, and all $A_v$'s commute with $W_S[P]$ we find
	\begin{equation}
		W_S[P] \left(  W_S[p] \big( \prod_{v \in p} A_v \big) \right) W_S[P]^* = W_S[P] V[Q] \left( \prod_{v \in p} A_v \right) = W_S[p] \left( \prod_{v \in p} A_v \right).
	\end{equation}
	The claim for semion string operators $W_S[P]$ follows.

	The required result is easy to verify for bound state strings $W_B[P]$, and since $W_{\bar S}[P] = W_{S}[P] W_{B}[P]$, the claim also holds for the anti-semion string operators.
\end{proof}

\begin{lemma} \label{lem:finite closed strings leave the ground state invariant}
	If $P$ is a finite closed string, then
	\begin{equation}
		\omega \circ w_a[P] = \omega
	\end{equation}
	for any $a \in \{1, S, \bar S, B\}$.
\end{lemma}

\begin{proof}
	By Lemmas \ref{lem:Omega_n is unique ground state} and \ref{lem:strings commute with Hamiltonian} we have
	\begin{equation}
		W_a[P] \Omega_n \sim \Omega_n
	\end{equation}
	for $n$ sufficiently large. Hence $\omega_n \circ w_a[P] = \omega_n$ for $n$ sufficiently large. The double semion state $\omega$ is by definition the weak-* limit of the sequence $\omega_n$ so
	\begin{equation}
		\omega \circ w_a[P] = \lim_{n \uparrow \infty} \, \omega_n \circ w_a[P] = \lim_{n \uparrow \infty} \omega_n = \omega.
	\end{equation}
\end{proof}

We are now ready to give the

\begin{proofof}[Proposition \ref{prop:closed strings leave the ground state invariant}]
	Let $O$ be a strictly local observable. Then we can find a finite closed loop $P'$ such that $w_a[P](O) = w_a[P'](O)$. From Lemma \ref{lem:finite closed strings leave the ground state invariant} we then find $(\omega \circ w_a[P])(O) = (\omega \circ w_a[P'])(O) = \omega(O)$. Since the strictly local operators are dense in $\caA$, it follows that $\omega \circ w_a[P] = \omega$.
\end{proofof}

%%%%%%%%%%%%%%%%%%%%%%%%%%%%%%%%%%%%%%%%%%%%%%%%%%%%%%%%%%%%%%%%%%%%%%%%%%%%%%%%%%%
%%%%%%%%%%%%%%%%%%%%%%%%%%%%%%%%%%%%%%%%%%%%%%%%%%%%%%%%%%%%%%%%%%%%%%%%%%%%%%%%%%%
%%%%%%%%%%%%%%%%%%%%%%%%%%%%%%%%%%%%%%%%%%%%%%%%%%%%%%%%%%%%%%%%%%%%%%%%%%%%%%%%%%%
%%%%%%%%%%%%%%%%%%%%%%%%%%%%%%%%%%%%%%%%%%%%%%%%%%%%%%%%%%%%%%%%%%%%%%%%%%%%%%%%%%%

\textbf{Acknowledgements}

A. B. was supported by the Simons Foundation.
B. K. was supported by the Villum Foundation through the QMATH Center of Excellence (Grant No. 10059) and the
Villum Young Investigator (Grant No. 25452) programs. 

%%%%%%%%%%%%%%%%%%%%%%%%%%%%%%%%%%%%%%%%%%%%%%%%%%%%%%%%%%%%%%%%%%%%%%%%%%%%%%%%%%%
%%%%%%%%%%%%%%%%%%%%%%%%%%%%%%%%%%%%%%%%%%%%%%%%%%%%%%%%%%%%%%%%%%%%%%%%%%%%%%%%%%%
%%%%%%%%%%%%%%%%%%%%%%%%%%%%%%%%%%%%%%%%%%%%%%%%%%%%%%%%%%%%%%%%%%%%%%%%%%%%%%%%%%%

\bibliographystyle{abbrvArXiv}
\bibliography{bib}

\begin{thebibliography}{10}

\bibitem{bachmann2012automorphic}
S.~Bachmann, S.~Michalakis, B.~Nachtergaele, and R.~Sims.
\newblock Automorphic equivalence within gapped phases of quantum lattice
  systems.
\newblock {\em Communications in Mathematical Physics}, 309(3):835--871, 2012.
\newblock  \href{https://arxiv.org/abs/1102.0842}{{\ttfamily arXiv:1102.0842}}.

\bibitem{bratteli2012operator}
O.~Bratteli and D.~W. Robinson.
\newblock {\em Operator algebras and quantum statistical mechanics: Volume 1:
  C*-and W*-Algebras. Symmetry Groups. Decomposition of States}.
\newblock Springer Science \& Business Media, 2012.

\bibitem{cha2020stability}
M.~Cha, P.~Naaijkens, and B.~Nachtergaele.
\newblock On the stability of charges in infinite quantum spin systems.
\newblock {\em Communications in Mathematical Physics}, 373(1):219--264, 2020.
\newblock  \href{https://arxiv.org/abs/1804.03203}{{\ttfamily
  arXiv:1804.03203}}.

\bibitem{dijkgraaf1991quasi}
R.~Dijkgraaf, V.~Pasquier, and P.~Roche.
\newblock Quasi hopf algebras, group cohomology and orbifold models.
\newblock {\em Nuclear Physics B-Proceedings Supplements}, 18(2):60--72, 1991.

\bibitem{doplicher1971local}
S.~Doplicher, R.~Haag, and J.~E. Roberts.
\newblock Local observables and particle statistics i.
\newblock {\em Communications in Mathematical Physics}, 23:199--230, 1971.

\bibitem{doplicher1974local}
S.~Doplicher, R.~Haag, and J.~E. Roberts.
\newblock Local observables and particle statistics ii.
\newblock {\em Communications in Mathematical Physics}, 35:49--85, 1974.

\bibitem{fredenhagen1989superselection}
K.~Fredenhagen, K.-H. Rehren, and B.~Schroer.
\newblock Superselection sectors with braid group statistics and exchange
  algebras: I. general theory.
\newblock {\em Communications in Mathematical Physics}, 125:201--226, 1989.

\bibitem{fredenhagen1992superselection}
K.~Fredenhagen, K.-H. Rehren, and B.~Schroer.
\newblock Superselection sectors with braid group statistics and exchange
  algebras ii: Geometric aspects and conformal covariance.
\newblock {\em Reviews in Mathematical Physics}, 4(spec01):113--157, 1992.

\bibitem{frohlich1990braid}
J.~Fr{\"o}hlich, F.~Gabbiani, and P.-A. Marchetti.
\newblock Braid statistics in three-dimensional local quantum theory.
\newblock {\em Physics, Geometry and Topology}, pages 15--79, 1990.

\bibitem{glimm1960certain}
J.~G. Glimm.
\newblock On a certain class of operator algebras.
\newblock {\em Transactions of the American Mathematical Society},
  95(2):318--340, 1960.

\bibitem{kawagoe2020microscopic}
K.~Kawagoe and M.~Levin.
\newblock Microscopic definitions of anyon data.
\newblock {\em Physical Review B}, 101(11):115113, 2020.
\newblock  \href{https://arxiv.org/abs/1910.11353}{{\ttfamily
  arXiv:1910.11353}}.

\bibitem{kitaev2003fault}
A.~Y. Kitaev.
\newblock Fault-tolerant quantum computation by anyons.
\newblock {\em Annals of physics}, 303(1):2--30, 2003.
\newblock  \href{https://arxiv.org/abs/quant-ph/9707021}{{\ttfamily
  arXiv:quant-ph/9707021}}.

\bibitem{levin2005string}
M.~A. Levin and X.-G. Wen.
\newblock String-net condensation: A physical mechanism for topological phases.
\newblock {\em Physical Review B}, 71(4):045110, 2005.
\newblock  \href{https://arxiv.org/abs/cond-mat/0404617}{{\ttfamily
  arXiv:cond-mat/0404617}}.

\bibitem{naaijkens2011localized}
P.~Naaijkens.
\newblock Localized endomorphisms in kitaev's toric code on the plane.
\newblock {\em Reviews in Mathematical Physics}, 23(04):347--373, 2011.
\newblock  \href{https://arxiv.org/abs/1012.3857}{{\ttfamily arXiv:1012.3857}}.

\bibitem{ogata2022derivation}
Y.~Ogata.
\newblock A derivation of braided c*-tensor categories from gapped ground
  states satisfying the approximate haag duality.
\newblock {\em Journal of Mathematical Physics}, 63(1):011902, 2022.
\newblock  \href{https://arxiv.org/abs/arxiv:2106.15741}{{\ttfamily
  arxiv:2106.15741}}.

\bibitem{propitius1995topological}
M.~d.~W. Propitius.
\newblock Topological interactions in broken gauge theories.
\newblock {\em arXiv preprint hep-th/9511195}, 1995.
\newblock  \href{https://arxiv.org/abs/hep-th/9511195}{{\ttfamily
  arXiv:hep-th/9511195}}.

\end{thebibliography}

\end{document}